\documentclass[11pt]{article}

\usepackage{hdb_macros}
\usepackage{fullpage}
\usepackage{xcolor}
\usepackage{thm-restate}

\renewcommand{\C}{\mathcal{C}}
\newcommand{\D}{\mathcal{D}}
\newcommand{\GapCDP}{\problem{GapCDP}}
\newcommand{\GapNorm}{\problem{GapNorm}}
\newcommand{\CDP}{\problem{CDP}}

\newcommand{\LCE}{\problem{LCE}}
\newcommand{\LIP}{\problem{LIP}}

\newcommand{\supp}{\mathrm{supp}}
\newcommand{\defeq}{\coloneqq}
\newcommand{\zznorm}[1]{\norm{#1}_{0 \to 0}}
\newcommand{\M}{\mathcal{M}}

\newcommand{\full}[1]{#1}
\newcommand{\fullornot}[2]{#1}

\newcounter{casenum}
\newenvironment{caseof}{\setcounter{casenum}{1}}{\vskip.5\baselineskip}
\newcommand{\case}[2]{\vskip.5\baselineskip\par\noindent {\bfseries Case \arabic{casenum}:} #1\\#2\addtocounter{casenum}{1}}

\title{The Code Distortion Problem}
\author{\fullornot{Huck Bennett\thanks{University of Colorado Boulder. \email{Huck.Bennett@colorado.edu}. Supported in part by NSF Award No. 2432132.} \and Matthew Fox\thanks{University of Colorado Boulder. \email{Matthew.Fox@colorado.edu}.} \and Bryant Morrell\thanks{University of Colorado Boulder. \email{Bryant.Morrell@colorado.edu}. Supported in part by NSF Award No. 2432132.}}{(Anonymous.)}}
\date{\today}

\begin{document}

\maketitle
\begin{abstract}
Two linear error-correcting codes $\C_1, \C_2 \subseteq \F_q^n$ are called \emph{linearly equivalent} if there is a linear isometry mapping $\C_1$ to $\C_2$.
In this work, we generalize the notion of linear equivalence and study the minimum distortion $\D(\C_1, \C_2)$ of a linear mapping between codes $\C_1, \C_2 \subseteq \F_q^n$, which \emph{quantifies} how similar $\C_1$ and $\C_2$ are. 
We introduce and study the Code Distortion Problem ($\CDP$), which asks to find a minimum distortion mapping between two input codes $\C_1$ and $\C_2$.
$\CDP$ generalizes the Linear Code Equivalence Problem ($\LCE$), which is essentially the special case of $\CDP$ where $\D(\C_1, C_2) = 1$ and which is well-studied because of its role in cryptography.

We prove that (decisional) $\CDP$ is $\NP$-hard to approximate to within any constant factor, and that it is in $\Sigma_2^\P$.
We also give a single-exponential-time $k^2$-approximation algorithm for $\CDP$, where $k$ is the dimension of the input codes. Furthermore, we give a single-exponential-time $\big(\frac{2k + 1}{3})^2$-approximation algorithm for a natural special case of $\CDP$, and we show that our analysis is tight in this case.

We use techniques from analogous work on the Lattice Distortion Problem ($\LDP$) by Bennett, Dadush, and Stephens-Davidowitz (ESA, 2016). We also introduce or study a number of additional concepts that might be of independent interest. These include an adaptation of the celebrated reduction of Goldreich, Micciancio, Safra, and Seifert (IPL, 1999) from the Shortest Vector Problem ($\SVP$) to the Closest Vector Problem ($\CVP$) on lattices to the analogous problems on codes; successive minima bases for codes; and the matrix $0 \to 0$ ``norm'' on subspaces.

\end{abstract}

\newpage

\section{Introduction}
\label{sec:introduction}

A \emph{linear error-correcting code} (or simply, \emph{code}) is a linear subspace $\C \subseteq \F_q^n$.
Although best known for their use in robust communication, codes also appear prominently in algorithms, computational complexity, and cryptography. Notably, several of the cryptosystems appearing in late rounds or standardized by the National Institute of Standards and Technology's (NIST) post-quantum cryptography standardization process are based on the intractability of certain computational problems on codes~\cite{classic-mceliece,conf/africacrypt/BiasseMPS20,hqc_spec_2025}. 

One central problem on codes of cryptographic interest is the \emph{Linear Code Equivalence Problem} ($\LCE$), in which, given generator matrices of two codes $\C_1, \C_2 \subseteq \F_q^n$ as input, the goal is to decide whether $\C_1$ and $\C_2$ are \emph{linearly equivalent}, i.e., whether one can convert $\C_1$ into $\C_2$ simply by permuting coordinates and scaling coordinates by non-zero values.
In particular, the LESS cryptosystem~\cite{conf/africacrypt/BiasseMPS20,conf/pqcrypto/BarenghiBPS21} crucially relies on the hardness of $\LCE$. It is currently a second-round candidate in NIST's ``Additional Digital Signature Schemes'' standardization process~\cite{nist_pqc_round2_additional_signatures_2025}.
Furthermore, a large body of work has studied $\LCE$ and its variants from an algorithmic and complexity-theoretic standpoint. See, e.g.,~\cite{Leon82,article/PetrankR97,journals/tit/Sendrier00,conf/soda/BabaiCGQ11,conf/sacrypt/Beullens20,journals/amco/BarenghiBPS23,conf/pkc/DucasG23,BennettWin2024,Nowakowski25,BBB+asymp-improvements26}.

A linear isometry (i.e., distance-preserving map) on $\F_q^n$ corresponds to a matrix $M \in \F_q^{n \times n}$ such that for all $\vec{x} \in \F_q^n$, $\norm{M \vec{x}}_0 = \norm{\vec{x}}_0$, where $\norm{\cdot}_0$ denotes the Hamming weight of a vector.
It is not hard to check that a matrix $M$ satisfying this property can be written as the product $M = DP$ of a non-singular diagonal matrix $D$ (which corresponds to scaling) and a permutation matrix $P$ (such matrices $M$ are called \emph{monomial matrices}).
In other words, the two operations in the definition of $\LCE$---permuting and scaling coordinates by non-zero values---exactly characterize the linear isometries on $\F_q^n$.

So, $\LCE$ asks whether there is a (perfectly) distance-preserving linear map between two input codes $\C_1, \C_2$. However, it is natural to ask a refined version of this question: %
\begin{quote}
    Is there a linear map $T$ such that $T(\C_1) = \C_2$ and $T$ \emph{approximately} preserves distances?
\end{quote}
We make this precise and quantitative by asking what the minimum \emph{distortion} of such a map $T$ from $\C_1$ to $\C_2$ is.
We define the distortion of a map $T$ with $T(\C_1) = \C_2$ as
\begin{equation} \label{eq:distortion-def-intro}
\D_T(\C_1, \C_2) :=
\Big( \max_{\vec{x} \in \C_1 \setminus \set{\vec{0}}} \frac{\norm{T\vec{x}}_0}{\norm{\vec{x}}_0} \Big) \big/ \Big(\min_{\vec{x} \in \C_1\setminus \set{\vec{0}}} \frac{\norm{T \vec{x}}_0}{\norm{\vec{x}}_0} \Big) \ \text{,}
\end{equation}
and we define $\D(\C_1, \C_2) := \min_{T} \D_T(\C_1, \C_2)$, where the minimum is taken over all linear maps $T$ such that $T(\C_1) = \C_2$.
Informally, $\D_T(\C_1, \C_2)$ is the ratio of the most $T$ blows up distances and the most it shrinks distances.
We note that $\D_T(\C_1, \C_2)$ is simply the distortion of the map $T$ in the standard sense of metric embeddings, where the codes $\C_1$ and $\C_2$ are the ambient spaces and each is equipped with the Hamming metric. (See, e.g., the lecture notes of Matou{\v{s}}ek~\cite{Matousek2013MetricEmbeddings} for background on metric embeddings.)
Moreover, $\D(\C_1, \C_2)$ is the minimum distortion over all such linear maps $T$, and it therefore quantifies how similar $\C_1$ and $\C_2$ are. Intuitively, linearly equivalent codes are essentially the same, and codes $\C_1, \C_2$ with $\D(\C_1, \C_2) \gg 1$ are quite different.
One can show that $\D(\C_1, \C_2) \geq 1$ (see \cref{lem:code-dist-basic-bounds}), and that equality holds if and only if there exist linearly equivalent ``unary scalings'' of $\C_1$ and $\C_2$ (see \cref{prop:equiv-codes-scaling-up}).  

Recall that a code $\C \subseteq \F_q^n$ of dimension $k$, called an $[n, k]_q$ code, is often represented by a generator matrix (basis) $G \in \F_q^{n \times k}$.\footnote{In this work we use \emph{column} bases for codes, although row bases are often used in the literature.} 
I.e., the code generated by $G$ is $\C(G) := \set{G\vec{x} : \vec{x} \in \F_q^k}$.
And, one can represent an arbitrary linear map between two $[n, k]_q$ codes in terms of generator matrices $G_1 := (\vec{g}_1^{(1)}, \ldots, \vec{g}_k^{(1)}) \in \F_q^{n \times k}$ and $G_2 := (\vec{g}_1^{(2)}, \ldots, \vec{g}_k^{(2)}) \in \F_q^{n \times k}$ of the respective codes by $\vec{g}_i^{(1)} \mapsto \vec{g}_i^{(2)}$.

For example, consider the generator matrices
\[
G_1 :=
\begin{pmatrix}
1 & 0 \\
0 & 1 \\
0 & 1 \\
0 & 1
\end{pmatrix} \in \F_2^{4 \times 2} \ , \qquad
G_2 :=
\begin{pmatrix}
1 & 0 \\
1 & 1 \\
1 & 1 \\
0 & 1
\end{pmatrix} \in \F_2^{4 \times 2} \ , \qquad
G_2' :=
\begin{pmatrix}
1 & 0 \\
0 & 1 \\
0 & 1 \\
1 & 1
\end{pmatrix} \in \F_2^{4 \times 2} \ \text{.}
\]
Notice that $\C_1 := \C(G_1)$ and $\C_2 := \C(G_2)$ are not linearly equivalent, but that $\C_2 = \C(G_2')$, i.e., $G_2$ and $G_2'$ are different generator matrices of the same code.
Letting $T$ and $T'$ be linear maps such that $T(G_1) = G_2$ and $T'(G_1) = G_2'$, it is straightforward to check that
\begin{align*}
\D_T(\C_1, \C_2) 
&= \Big( \max_{\vec{x} \in \C_1 \setminus \set{\vec{0}}} \frac{\norm{T\vec{x}}_0}{\norm{\vec{x}}_0} \Big) \big/ \Big(\min_{\vec{x} \in \C_1\setminus \set{\vec{0}}} \frac{\norm{T \vec{x}}_0}{\norm{\vec{x}}_0} \Big) \\
&= \Big(\frac{\norm{(1, 1, 1, 0)}_0}{\norm{(1, 0, 0, 0)}_0}\Big) \big/ \Big(\frac{\norm{(1, 0, 0, 1)}_0}{\norm{(1, 1, 1, 1)}_0}\Big) \\
&= \frac{3}{1} \cdot \Big(\frac{1}{2}\Big)^{-1} = 6 \ \text{,}
\end{align*}
and
\begin{align*}
\D_{T'}(\C_1, \C_2) 
&= \Big( \max_{\vec{x} \in \C_1 \setminus \set{\vec{0}}} \frac{\norm{T' \vec{x}}_0}{\norm{\vec{x}}_0} \Big) \big/ \Big(\min_{\vec{x} \in \C_1\setminus \set{\vec{0}}} \frac{\norm{T' \vec{x}}_0}{\norm{\vec{x}}_0} \Big) \\
&= \Big(\frac{\norm{(1, 0, 0, 1)}_0}{\norm{(1, 0, 0, 0)}_0}\Big) \big/ \Big(\frac{\norm{(1, 1, 1, 0)}_0}{\norm{(1, 1, 1, 1)}_0}\Big) \\
&= \frac{2}{1} \cdot \Big(\frac{3}{4}\Big)^{-1} = \frac{8}{3} \ \text{.}
\end{align*}
In other words, $T'$ is a substantially lower distortion mapping between $\C_1$ and $\C_2$ than $T$ is.
Part of the reason for this is that $G_2'$ is a better basis of $\C_2$ than $G_2$ is in a sense that we will make precise in the sequel.

In this work, we study code distortion from a computational standpoint. Let $\gamma \geq 1$ be an approximation factor. Specifically, we study the $\gamma$-approximate Code Distortion Problem ($\gamma$-$\CDP$), which asks to compute a linear mapping $T$ between two codes $\C_1$ and $\C_2$ such that $\D_T(\C_1, \C_2) \leq \gamma \cdot \D(\C_1, \C_2)$. We also study its decision version, $\gamma$-$\GapCDP$, in which the goal is to decide whether $\D(\C_1, \C_2) \leq D$ or $\D(\C_1, \C_2) > \gamma D$ for a given input distortion threshold value $D \geq 1$.
When emphasizing the finite field $\F_q$ that the input codes are over, we add a subscript $q$ to the problems.

\subsection{Our Results}
\label{sec:intro-results}

We study $\CDP$ from both an algorithmic and complexity-theoretic perspective.
We first prove hardness of approximation for $\GapCDP$ with (fixed) $D > 1$.
\begin{restatable}{theorem}{hardnessapprox} \label{thm:intro-hardness-of-approx}
For any constant $\gamma \geq 1$, any constant $D > 1$, and any prime power $q \leq \poly(n)$, $\gamma$-$\GapCDP_q$ with distortion threshold $D$ is $\NP$-hard under deterministic Cook reductions.
\end{restatable}

As we formalize in \cref{prop:equiv-codes-scaling-up}, $\GapCDP$ with distortion threshold $D = 1$ is equivalent to $\LCE$ up to scaling. From this it follows that $\GapCDP$ with $D = 1$ is in $\NP$, but, because of a result of Petrank and Roth~\cite{article/PetrankR97} showing that $\LCE$ is in $\coAM$, it is not $\NP$-hard unless the polynomial hierarchy collapses.
In contrast, \cref{thm:intro-hardness-of-approx} asserts that if $D$ is just slightly larger than $1$ then $\GapCDP$ is $\NP$-hard even to approximate to within any constant factor $\gamma$.
For example, taking $D = 1.01$ and $\gamma = 1000$, \cref{thm:intro-hardness-of-approx} shows that it is $\NP$-hard to decide whether $\D(\C_1, \C_2) \leq 1.01$ or $\D(\C_1, \C_2) > 1010$ when one of these is promised to hold.
The result holds for codes $\C_1, \C_2 \subseteq \F_q^n$ for any field $\F_q$ with $q \leq \poly(n)$.

We further show in \cref{thm:sigma2-containment} that $\GapCDP$ is in the complexity class $\Sigma_2^\P$.
Interestingly, it is not clear that $\GapCDP$ is contained in $\NP$. 
This stands in contrast both to $\LCE$ and the Lattice Distortion Problem ($\LDP$), the analogous problem to $\GapCDP$ on lattices, which~\cite{conf/esa/BennettDS16} noted is in $\NP$.%
\footnote{Recall that a \emph{lattice} is the analog of a code over the integers or real numbers. Specifically, the lattice generated by a basis $B \in \R^{n \times k}$ with full column rank is $\lat(B) := \set{B\vec{z} : \vec{z} \in \Z^k}$.}
The main difference between $\LDP$ and $\GapCDP$ is that the $2 \to 2$ norm $\norm{T}$ of a real-valued matrix $T$ is efficient to compute (indeed, $\norm{T}$ is simply the largest singular value of $T$), whereas the problem of computing the $0 \to 0$ ``norm'' of a linear operator $T$ mapping one code to another is $\coNP$-hard even to approximate to within any constant factor! See \cref{cor:zznorm-conp-hard}.\footnote{The $0 \to 0$ norm of a linear operator acting on all of $\F_q^n$ (and not restricted restricted to a code/subspace) is also efficiently computable; see \cref{prop:zznorm-val}.}

We next give algorithms for $\CDP$. 
In \cref{lem:brute-force-alg}, we give and analyze a brute force, exact algorithm for $\CDP_q$ on $[n, k]_q$ codes, which runs in $O^*(q^{k(k+1)}) \approx q^{k^2}$ time---essentially the amount of time it takes to enumerate all generator matrices of a given $[n, k]_q$ code.\footnote{In this paper, we use $O^*(\cdot)$ to suppress polynomial factors in $n$ and $\log q$.} It is unclear how to get a faster exact algorithm, and so we turn to giving approximation algorithms.
We first give a single-exponential-time, $k^2$-approximation algorithm for $\CDP$ on $[n, k]_q$ codes.
\begin{restatable}{theorem}{approxalg} \label{thm:intro-approx-alg}
For any prime power $q$, there is an algorithm for $k^2$-$\CDP$ on $[n, k]_q$ codes that runs in $O^*(q^k)$ time and $\poly(n, \log q)$ space.
\end{restatable}

We contrast the approximation factor of $k^2$ in \cref{thm:intro-approx-alg} with the trivial bound bound of $\D_T(\C_1, \C_2) \leq n^2$, which holds for any linear map $T$ such that $T(\C_1) = \C_2$ for $[n, k]_q$ codes $\C_1, \C_2$; see \cref{lem:code-dist-basic-bounds}. Since $\D_T(\C_1, \C_2) \geq 1$ also always holds, this implies a trivial $n^2$-approximation algorithm for $\GapCDP$.
We do not know of a better bound (even a non-constructive one), although we show an upper bound of $n$ when allowing for \emph{non-linear} maps between the codes; see \cref{lem:non-linear-dist-ub}.
However, even in this case it is not clear that \emph{finding} such a map would be at all efficient, and we note that \cref{thm:intro-approx-alg} holds for approximate \emph{search} $\CDP$.

For an $[n, k]_q$ code $\C$ and $i \in [k]$, we define the \emph{$i$th successive minimum} $\lambda_i(\C)$ of $\C$ to be the minimum value $r \in \set{1, \ldots, n}$ such that $\C$ contains $i$ linearly independent vectors of Hamming weight at most $r$. In particular, $\lambda_1(\C)$ is the minimum distance of $\C$.
We get a strengthening of \cref{thm:intro-approx-alg} when the input codes are binary and when all of the successive minima of both of the input codes are the same.\footnote{We note that all of the successive minima of a code being the same is equivalent to the code being generated by its minimum-weight non-zero codewords.}

\begin{restatable}{theorem}{approxalgsuccmin} \label{thm:intro-approx-alg-succ-min-same}
There is an algorithm for $\big(\frac{2k + 1}{3}\big)^2$-$\CDP$ on $[n, k]_2$ codes $\C_1, \C_2$ satisfying $\lambda_1(\C_1) = \cdots = \lambda_k(\C_1) = \lambda_1(\C_2) = \cdots = \lambda_k(\C_2)$ that runs in $O^*(2^k)$ time and $\poly(n)$ space.
\end{restatable}

In fact, the \emph{algorithm} used to prove both \cref{thm:intro-approx-alg,thm:intro-approx-alg-succ-min-same} is the same, but all of the codes' successive minima being the same allows for a tighter analysis of its approximation factor.
\fullornot{
Furthermore, in \cref{thm:tight-approx} we show that our analysis of this algorithm is tight in the setting of \cref{thm:intro-approx-alg-succ-min-same}.
Although we give \cref{thm:intro-approx-alg-succ-min-same} only for $\F_2$, one can show that it extends to $\F_q$ for arbitrary prime powers $q$ with a more complicated proof; see~\cite{morrell2026adventures}.}
{In fact, the analysis of \cref{thm:intro-approx-alg-succ-min-same} is tight. However, due to space constraints, we defer the proof to the full version of our paper. Additionally, we note that although we give \cref{thm:intro-approx-alg-succ-min-same} only for $\F_2$, one can show that it extends to $\F_q$ for arbitrary prime powers $q$ with a more complicated proof; see~\cite{morrell2026adventures}.}

We note that there are a number of natural $[n, k]_q$ codes $\C$ with $\lambda_1(\C) = \lambda_k(\C)$. For example, $\C := \F_2^n$, the binary code $\C := \set{\vec{x} \in \F_2^n : \iprod{\vec{x}, \vec{1}} \equiv 0 \bmod 2}$ defined by the single parity-check constraint $\vec{1}$, and Reed-Solomon codes $\C \subseteq \F_q^n$ all have this property.\footnote{Recall that the codewords in a Reed-Solomon code $\C \subseteq \F_q^n$ of dimension $k$ are evaluations of polynomials $p(x) \in \F_q[x]$, $\deg(p) < k$ on $n$ distinct elements $\alpha_1, \ldots, \alpha_n \in \F_q$. It is not hard to show that codewords $p_S(\alpha_i)_{i=1}^n$ for polynomials $p_S(x) := \prod_{s \in S} (x - s)$ with $S \subset \set{\alpha_1, \ldots, \alpha_n}$, $\card{S} = k - 1$ are minimum-weight non-zero codewords in such a Reed-Solomon code $\C$. Furthermore, one can show that there are $k$ linearly independent such codewords by considering the codewords induced by $p_{S_i}$ for $S_i := \set{\alpha_1, \ldots, \alpha_k} \setminus \set{\alpha_i}$ for $i \in [k]$.}
Additionally, we construct an $[n, k]_2$ code $\C \subseteq \F_2^n$ with $\lambda_1(\C) = \lambda_k(\C)$ in \cref{thm:tight-approx} to show the tightness of \cref{thm:intro-approx-alg-succ-min-same}.

\subsection{Overview of Techniques}
\label{sec:intro-overview}

At a high level, the proofs of \cref{thm:intro-hardness-of-approx,thm:intro-approx-alg} follow along similar lines to those used to show upper and lower bounds for the Lattice Distortion Problem ($\LDP$) in~\cite{conf/esa/BennettDS16}.
However, they also require a number of additional concepts, which may be of independent interest.
The proof of \cref{thm:intro-approx-alg-succ-min-same} goes along similar lines to the proof of \cref{thm:intro-approx-alg}, but it is substantially more nuanced.

\paragraph{Approximation algorithms.}
We prove a constructive upper bound using what we call \emph{successive minima bases} for codes.
We define a successive minima basis of an $[n, k]_q$ code $\C$ to be a generator matrix $G := (\vec{g}_1, \ldots, \vec{g}_k) \in \F_q^{n \times k}$ such that $\norm{\vec{g}_i}_0 = \lambda_i(\C)$, i.e., the columns of $G$ are linearly independent and their Hamming weights achieve the successive minima of $\C$.
Every code has a successive minima basis (essentially because any $k$ linearly independent vectors in $\C$ form a basis of $\C$) and they can be computed in roughly $q^k$ time; see \cref{lem:succ-min-basis}.
In contrast, not every lattice $\lat$ has a successive minima basis.%
\footnote{\label{foot:lattice-succ-min-counter}As an explicit example, consider the lattice $\lat = \lat(B) \subset \R^5$ generated by the basis
\[
B = (\vec{b}_1, \ldots, \vec{b}_5) := 
\begin{pmatrix}
1 & 0 & 0 & 0 & 1/2 \\
0 & 1 & 0 & 0 & 1/2 \\
0 & 0 & 1 & 0 & 1/2 \\
0 & 0 & 0 & 1 & 1/2 \\
0 & 0 & 0 & 0 & 1/2
\end{pmatrix} \ \text{.}
\]
One can check that $\vec{e}_i \in \lat$ for $i = 1, \ldots, 5$, and that every vector in $\lat$ has norm at least $1$. So, $\lambda_1(\lat) = \cdots = \lambda_5(\lat) = 1$. However, any basis of $\lat$ must contain $\pm \vec{b}_5$, which has norm greater than $1$.}
Although successive minima are widely used in the study of lattices, to the best of our knowledge they have not previously been explicitly used for codes.

Our approximation algorithm for $\CDP$ that leads to \cref{thm:intro-approx-alg,thm:intro-approx-alg-succ-min-same} is simply to compute successive minima bases $G_1 := (\vec{g}_1^{(1)}, \ldots, \vec{g}_k^{(1)}) \in \F_q^{n \times k}$ and $G_2 := (\vec{g}_1^{(2)}, \ldots, \vec{g}_k^{(2)}) \in \F_q^{n \times k}$ of the respective input $[n, k]_q$ codes $\C_1$ and $\C_2$, and to take the mapping $T : \vec{g}_i^{(1)} \mapsto \vec{g}_i^{(2)}$ that they induce.

We show that such a mapping $T$ is a $k^2$-approximation, proving \cref{thm:intro-approx-alg}, as follows.
For $[n, k]_q$ codes $\C_1$ and $\C_2$, define $M(\C_1, \C_2) := \max_{i \in [k]} \lambda_i(\C_1)/\lambda_i(\C_2)$. (An analogous quantity for lattices appeared in~\cite{conf/esa/BennettDS16}.)
We show that the mapping $T$ defined above in terms of successive minima bases satisfies $\D_T(\C_1, \C_2) \leq k^2 \cdot M(\C_1, \C_2) M(\C_2, \C_1)$.
Furthermore, we show that \emph{any} linear map $T'$ such that $T' \C_1 = \C_2$ must satisfy $\D_{T'}(\C_1, \C_2) \geq M(\C_1, \C_2) M(\C_2, \C_1)$, and therefore $\D(\C_1, \C_2) \geq M(\C_1, \C_2) M(\C_2, \C_1)$.
Combining these two bounds, we have that
\[
M(\C_1, \C_2) M(\C_2, \C_1) \leq \D(\C_1, \C_2) \leq \D_T(\C_1, \C_2) \leq k^2 \cdot M(\C_1, \C_2) M(\C_2, \C_1) \ \text{,}
\]
and therefore that $T$ is $k^2$-approximately optimal, i.e., that $\D_T(\C_1, \C_2) \leq k^2 \cdot \D(\C_1, \C_2)$.

Showing the refined result in \cref{thm:intro-approx-alg-succ-min-same} for codes over $\F_2$ all of whose successive minima are the same requires more work.
We assume that $G_1$ and $G_2$ are successive minima bases of the input $[n, k]_2$ codes (so that each vector in $G_1$ and $G_2$ has Hamming weight $\lambda := \lambda_1(\C_1) = \cdots = \lambda_k(\C_1) = \lambda_1(\C_2) = \cdots = \lambda_k(\C_2)$), and we let $T$ be a mapping such that $TG_1 = G_2$. We then analyze
\[
\D_T(\C_1,\C_2) := \max_{\vec{a} \in \F_2^k}\Big(\frac{\norm{G_2\vec{a}}_0}{\norm{G_1\vec{a}}_0}\Big)\max_{\vec{b} \in \F_2^k}\Big(\frac{\norm{G_1\vec{b}}_0}{\norm{G_2\vec{b}}_0}\Big) = \max_{\vec{a},\vec{b} \in \F_2^k}\Big(\frac{\norm{G_1\vec{b}}_0}{\norm{G_1\vec{a}}_0}\frac{\norm{G_2\vec{a}}_0}{\norm{G_2\vec{b}}_0}\Big) 
\]
We upper bound the quantity $\frac{\norm{G_2\vec{a}}_0\norm{G_1\vec{b}}_0}{\norm{G_1\vec{a}}_0\norm{G_2\vec{b}}_0}$ on the right-hand side for non-zero $\vec{a}, \vec{b} \in \F_q^k$ in cases.
In each case, we make use of the fact that for non-zero $\vec{x}$, $\lambda \leq \norm{G_1 \vec{x}}_0 \leq \lambda \norm{\vec{x}}_0$, where the upper bound follows from the triangle inequality. Similar bounds hold for $G_2$.

If $\norm{\vec{a}}_0 + \norm{\vec{b}}_0 \leq \big(\frac{2k + 1}{3}\big)$, then by the AM-GM inequality $\norm{\vec{a}}_0 \norm{\vec{b}}_0 \leq \big(\frac{2k + 1}{3}\big)^2 \lambda^2$.
So,
\[
\frac{\norm{G_2\vec{a}}_0\norm{G_1\vec{b}}_0}{\norm{G_1\vec{a}}_0\norm{G_2\vec{b}}_0} \leq \frac{\big(\frac{2k + 1}{3})^2 \lambda^2}{\lambda^2} = \Big(\frac{2k + 1}{3}\Big)^2 \ \text{.}
\]
On the other hand, if $\norm{\vec{a}}_0 + \norm{\vec{b}}_0 > \big(\frac{2k + 1}{3} \big)$ then the supports of $\vec{a}$ and $\vec{b}$ have large intersection, which leads to a similar upper bound on $\frac{\norm{G_2\vec{a}}_0\norm{G_1\vec{b}}_0}{\norm{G_1\vec{a}}_0\norm{G_2\vec{b}}_0}$.

\paragraph{Hardness of approximation.}
To show hardness of approximation for $\GapCDP$, we roughly follow the approach used in~\cite{conf/esa/BennettDS16} for showing hardness of approximation for $\LDP$.
We give a pair of reductions, the first of which reduces from the Minimum Distance Problem ($\MDP$) to a promise variant of the Nearest Codeword Problem ($\NCP$).
This first reduction is an adaptation of the reduction from the Shortest Vector Problem ($\GapSVP$) to the Closest Vector Problem ($\GapCVP$) on lattices from~\cite{journals/ipl/GoldreichMSS99}, and it is likely of independent interest.\footnote{We note that $\MDP$ and $\NCP$ are the coding problems analogous to the lattice problems $\GapSVP$ and $\GapCVP$, respectively. Furthermore, we note that while~\cite{journals/jcss/AroraBSS97} showed hardness of approximation for $\NCP$, it did not show hardness of the variant that we use in the present work, which has an added promise in the NO case. Finally, we note that our reduction has the advantage of being directly from $\MDP$, and of being dimension- and approximation-preserving.}

The second reduction is from $\NCP$ to $\GapCDP$. We recall that the input to $\NCP$ is a pair consisting of a generator matrix $G \in \F_q^{n \times k}$ of a code $\C$ and a target vector $\vec{t} \in \F_q^n$, and the goal is to decide whether the minimum distance between $\vec{t}$ and a codeword in $\C$ (i.e.,
$\min_{\vec{x} \in \C}\|\vec{t} - \vec{x}\|_0$) is at most some threshold.
The idea behind our hardness reduction is to construct the following pair of generator matrices, corresponding to an instances of $\CDP$:
\begin{equation} \label{eq:intro-ncp-to-cdp-attempt}
G_1 := \begin{pmatrix}
G & \vec{0} \\
0 & \vec{1}_r
\end{pmatrix} \ \text{,} \qquad
G_2 := \begin{pmatrix}
G & -\vec{t} \\
0 & \vec{1}_r
\end{pmatrix} \ \text{.}
\end{equation}
Here $r \in \Z^+$. Let $\C_1 := \C(G_1)$ and $\C_2 := \C(G_2)$.

If $\min_{\vec{c} \in \C}\|\vec{c} - \vec{t}\|_0$ is small then $\D(\C_1, \C_2)$ is small.
If on the other hand $\min_{\vec{c} \in \C}\|\vec{c} - \vec{t}\|_0$ is large, then we would \emph{hope} is that $\D(\C_1, \C_2)$ is large.
However, this does not clearly work directly, and we need to modify the reduction.
(The attempted reduction sketched here in \cref{eq:intro-ncp-to-cdp-attempt} is the natural analog of the analogous hardness reduction in~\cite{conf/esa/BennettDS16} for $\LDP$; the modified reduction is slightly more involved.)

Because $\MDP$ is known to be $\NP$-hard to approximate to within any constant~\cite{journals/tit/DumerMS03,journals/tit/ChengW12}, this pair of reductions implies that $\CDP$ is also $\NP$-hard to approximate to within any constant.

\paragraph{The complexity of $\GapCDP$ and its relationship to the $0 \to 0$ norm.}
We note that the complexity of $\GapCDP$ seems closely related to the complexity of computing what we call the matrix \emph{$0 \to 0$ norm restricted to a subspace}.
The matrix $0 \to 0$ norm%
\footnote{In fact, $\norm{\cdot}_{0 \to 0}$ is not a matrix norm just as Hamming weight is not a vector norm (even over the real numbers).
It does not satisfy the absolute homogeneity property $\norm{\alpha M}_{0 \to 0} = \abs{\alpha} \norm{M}_{0 \to 0}$ for real-valued matrices $M$ and scalars $\alpha$, and the magnitude of a finite field element is not even defined.
However, the $0 \to 0$ norm satisfies the other axioms of a matrix norm, including the triangle inequality ($\norm{A + B}_{0 \to 0} \leq \norm{A}_{0 \to 0} + \norm{B}_{0 \to 0}$), and it is sub-multiplicative (i.e., $\norm{AB}_{0 \to 0} \leq \norm{A}_{0 \to 0} \norm{B}_{0 \to 0}$). So, we abuse notation slightly and refer to $\norm{\cdot}_{0 \to 0}$ as a matrix norm.}
\[
\zznorm{T} := \max_{\vec{x} \in \F_q^n \setminus \set{\vec{0}}} \frac{\norm{M\vec{x}}_0}{\norm{\vec{x}}_0} \ \text{.}
\]
of a matrix $T \in \F_q^{m \times n}$ is also efficiently computable; see \cref{prop:zznorm-val}. However, the matrix $0 \to 0$ norm restricted to a subspace $\C \subseteq \F_q^n$ (i.e., code $\C$), defined as,
\[
\zznorm{T_{|\C}} := \max_{\vec{x} \in \C \setminus \set{\vec{0}}} \frac{\norm{T\vec{x}}_0}{\norm{\vec{x}}_0} \ \text{.}
\]
is $\coNP$-hard even to approximate; see \cref{cor:zznorm-conp-hard}.
We note that for a linear map $T$ with $T \C_1 = \C_2$ for codes $\C_1, \C_2$, $\D_T(\C_1, \C_2) = \zznorm{T_{|\C_1}} \cdot \zznorm{(T^{-1})_{|\C_2}}$.
The most obvious approach to showing that $\GapCDP$ is in $\NP$ would be to use a low-distortion map $T$ between the input codes as a witness, and to have the verifier check that $\norm{T_{|\C_1}}_{0 \rightarrow 0} \cdot \norm{(T^{-1})_{|\C_2}}_{0 \rightarrow 0} \leq D$. However, because computing the matrix $0 \to 0$ norm restricted to a subspace is $\coNP$-hard, this approach does not work.
(We note that this $\coNP$-hardness does not obviously translate to $\GapCDP$.)
As a result, the best upper bound that we get on $\GapCDP$ is to show containment in $\Sigma_2^{\P}$. (As we also note elsewhere, it is a tantalizing question whether one can improve our $\NP$-hardness result for $\GapCDP$ to $\Sigma_2^{\P}$-hardness, and therefore show that $\GapCDP$ is $\Sigma_2^{\P}$-complete.)

\subsection{An Approach for a Faster Approximation Algorithm}
\label{sec:related-work}

A natural direction for making our approximation algorithm in \cref{thm:intro-approx-alg} more efficient would be to use reduced (but not necessarily optimal) bases in place of successive minima bases.
Indeed, this is how~\cite{conf/esa/BennettDS16} achieves a time-approximation tradeoff in its algorithm for $\LDP$, which includes a polynomial-time approximation algorithm in one regime.
(On the other hand, our approximation algorithm in \cref{thm:intro-approx-alg} only runs in polynomial time for an extreme parameter regime---when $k = O(\log_q n)$.)
And, a recent line of work has studied basis reduction algorithms for codes~\cite{journals/tit/Debris-AlazardD22,conf/approx/GhentiyalaS24}, including analogs of basis reduction algorithms for lattices.
However, it is not clear how to leverage bases reduced in any of the senses they consider to get an approximation algorithm running in a reasonable amount of time and with a reasonable approximation factor.

\subsection{Open Questions}
\label{sec:open-questions}

Our work leaves open several interesting questions. On the complexity side, our results showing that $\GapCDP$ is both $\NP$-hard and in $\Sigma_2^\P$ begs the question of whether $\GapCDP$ is also $\Sigma_2^\P$-hard, and thus $\Sigma_2^\P$-complete. Moreover, as part of the inspiration for this work comes from the analogous Lattice Distortion Problem ($\LDP$)~\cite{conf/esa/BennettDS16}, it is natural to ask whether (approximate) $\CDP$ reduces to $\LDP$.
We note that a fairly simple reduction from $\LCE$ to the Lattice Isomorphism Problem ($\LIP$)---problems that essentially correspond to the distortion $D = 1$ cases of $\CDP$ and $\LDP$, respectively---was given in~\cite{regev14,BennettWin2024}. 
On the algorithmic side, it is natural to ask whether there is a better exact algorithm than the brute force, roughly $q^{k^2}$-time algorithm that we give, and whether there is a more efficient approximation algorithm for $\CDP$ than our roughly $q^k$. Finally, we ask given whether it is possible to build cryptography whose security guarantee rests on the hardness of approximate $\CDP$. This seems particularly natural given the role of $\LCE$ both as security assumption for the LESS cryptosystem~\cite{conf/africacrypt/BiasseMPS20}, and as a closely related problem to the McEliece cryptosystem~\cite{mceliece78,classic-mceliece}.\footnote{One can check that the public and private generator matrices in McEliece generate linearly equivalent codes. However, it is not clear that breaking McEliece reduces to solving $\LCE$ or vice-versa.}

\full{
\subsection{Acknowledgments}
We thank Alexander Golovnev and Noah Stephens-Davidowitz for useful conversations~\cite{golovnev-stephens-davidowitz-personal-comm-26}, and in particular for allowing us to include \cref{thm:MDPtoGapNormReduction} about the $\coNP$-hardness of approximating the matrix $0 \to 0$ norm on subspaces in this work.
We also thank the anonymous APPROX reviewers for helpful comments, and in particular for identifying an incorrect inequality (which is now fixed).
}

\section{Preliminaries}
\label{sec:prelims}

We define a \emph{monomial matrix} $M$ to be the product $M = DP$ of a non-singular diagonal matrix $D$ and a permutation matrix $P$. We denote the set of $n \times n$ monomial matrices over a field $\F$ as $\M_n(\F)$.

\subsection{Codes}
A \emph{linear code} is a linear subspace $\C \subseteq \F^n$, where $\F$ is field.  Generally, $\F$ is the finite field $\F_q$ for some prime power $q$, and when $q=2$ (as is common in practice) the code is a  \emph{binary code}. A linear code $\C$ is often characterized as an $[n,k]$ or $[n,k,d]_q$ code, where $q$ indicates the order of the field $\F_q$ that the code is taken over, $k$ is the dimension of the code, and $d$ is the minimum Hamming distance of $\C$, i.e. 
$$d := \min_{\substack{\vec{x},\vec{y} \in \C\\\vec{x} \neq \vec{y}}} \norm{\vec{x}-\vec{y}}_0 = \min_{\vec{x} \in \C\setminus\{\vec{0}\}}\norm{\vec{x}}_0.$$
We also call $n$ the \emph{block length} of the code.

An $[n,k]$ code $\C \subseteq \F^n$ is usually specified by a \emph{generator matrix} $G \in \F^{n \times k}$ with linearly independent columns.  The code is then the span of the columns of $G$, and any \emph{codeword} $\vec{x} \in \C$ can be expressed as $G\vec{m}$ for some \emph{message} $\vec{m} \in \F^k$.  Likewise, each message $\vec{m} \in \F^k$ corresponds to a unique codeword $G\vec{m} \in \C$.

We will also find it useful to refer to the \emph{support} of a vector $\vec{x} \in \F^k$.  Denoted $\supp(\vec{x})$, the support of $\vec{x}$ is the set of indices $i$ for which $\vec{x}_i \neq 0$.  The Hamming weight can then be expressed as $|\supp(\vec{x})|$, the size of the support.

\begin{definition}[Successive Minima for Codes]
    Let $n,k \in \Z^+, k \le n$, let $q$ be a prime power, and let $\C$ be an $[n,k]_q$ code. For $1 \le i \le k$, the \emph{$i$th successive minimum} of $\C$, denoted $\lambda_i(\C)$, is defined as the minimum value $r \in \Z^+$ such that $\C$ contains at least $i$ linearly independent codewords with Hamming weight at most $r$.
\end{definition}
We note that $\lambda_1(\C)$ is the minimum distance $d$ of $\C$.
We emphasize again that successive minima are well-studied for lattices, but do not seem to have been studied much for codes.
We will make use of the following definition, which formalizes one notion of optimal bases for codes.

\begin{definition}[Successive Minima Basis]
    Let $n,k \in \Z^+, k \le n$, let $q$ be a prime power, and let $\C$ be an $[n,k]_q$ code. A generator matrix $G = (\vec{v}_1,\dots,\vec{v}_k)$ of $\C$ is a \emph{successive minima basis} of $\C$ if $\norm{\vec{v}_i}_0 = \lambda_i(\C)$ for all $1 \le i \le k$.
\end{definition}

We next give an algorithm for computing successive minima bases, which in particular implies that such bases always exist.

\begin{lemma} \label{lem:succ-min-basis}
    Let $n,k \in \Z^+, k \le n$, let $q$ be a prime power, and let $\C$ be an $[n,k]_q$ code.  Then there exists a $O^*(q^k)$-time, polynomial space algorithm for computing a sucessive minima basis of $\C$. In particular, such a basis of $\C$ always exists.
\end{lemma}

\begin{proof}
    Let $\vec{v}_1 \in \C$ be a codeword that satisfies $\norm{\vec{v}_1} = \lambda_1(\C)$, and define $\vec{v}_i$ recursively for $1 < i \le k$ so that
    $$\vec{v}_i \in \argmin_{\vec{x} \in (\C \setminus \lspan(\vec{v}_{1},\dots,\vec{v}_{i-1}))} \norm{\vec{x}}_0 \ \text{.} $$
    We prove inductively that $\norm{\vec{v}_i}_0 = \lambda_i(\C)$. This is clearly true in the base case of $i = 1$ since $\vec{v}_1$ is a shortest non-zero codeword by definition. 
    For the inductive step with $i \geq 2$, assume that we can find $\vec{v}_1,\dots,\vec{v}_{i-1} \in \C$ satisfying $\norm{\vec{v}_j}_0 = \lambda_j(\C)$ for $1 \le j \le i-1$.  Additionally, observe that for any $1 \le i \le k$, the span of the set $S_i = \{\vec{v} \in \C : \norm{\vec{v}}_0 \le \lambda_i(\C)\}$ has dimension at least $i$ as a result of there being at least $i$ linearly independent codewords in $S_i$.  On the other hand, the set $\{\vec{v}_1,\dots,\vec{v}_{i-1}\} \subset S_i$ has dimension only $i-1$, so there must exist some $\vec{v}_i \in \C$ such that $\norm{\vec{v}_i}_0 \le \lambda_i(\C)$ and $\vec{v}_i \notin \lspan(\vec{v}_1, \ldots, \vec{v}_{i-1})$.
    
    We next show that $\norm{\vec{v}_i}_0 \ge \lambda_i(\C)$ (and therefore $\norm{\vec{v}_i}_0 = \lambda_i(\C)$) for $i \geq 2$. 
    Suppose not. Then $\norm{\vec{v}_i}_0 < \lambda_i(\C)$. 
    We have by definition that $\vec{v}_1, \ldots, \vec{v}_i$ are linearly independent, and by the induction hypothesis, $\lambda_j(\C) = \norm{\vec{v}_{j}}_0$ for all $1 \leq j \leq i - 1$.
    By definition, $\norm{\vec{v}_1}_0 \leq \cdots \leq \norm{\vec{v}_i}_0$, and it follows that $\vec{v}_1, \ldots, \vec{v}_i \in \C$ are linearly independent vectors all of Hamming weight strictly less than $\lambda_i(\C)$, which is a contradiction.

    The recursive process above yields $k$ linearly independent codewords $\vec{v}_1,\dots,\vec{v}_k$ with $\norm{\vec{v}_i}_0 = \lambda_i(\C)$ for all $i$.  Since $\C$ is a $k$-dimensional subspace of $\F_q^{n}$, any $k$ linearly independent vectors in $\C$ (and particularly our $\vec{v}_1,\dots,\vec{v}_k$) must span $\C$ and thus form a basis for it. 
    
    Computing each $\vec{v}_i$ involves a search over the $O(q^k)$ codewords in $\C$ with polynomial time required to check the linear independence and Hamming weight of each, so it takes $O^*(q^k)$ time to compute. 
    This enumeration can be done in polynomial space.
    Therefore, by repeating this $k$ times, we find the entire successive minima basis in $O^*(q^k)$ time and polynomial space.
\end{proof}

\subsection{Coding Problems}

Given a code, it is natural to ask what its minimum distance is. The Minimum Distance Problem formalizes this, and we give here the approximation version of the problem.
\begin{definition}
For $\gamma = \gamma(k) \geq 1$ and a prime power $q$, the \emph{decisional $\gamma$-approximate Minimum Distance Problem over $\F_q$} ($\gamma$-$\MDP_q$) is the decision problem defined as follows. An instance consists of (a generator matrix $G \in \F_q^{n \times k}$ of) a code $\C$ and an integer distance $0 \le d \le m$. It is a:
\begin{itemize}
    \item YES instance if $\lambda_1(\C) \leq d$.
    \item NO instance if $\lambda_1(\C) > \gamma d$.
\end{itemize}
\end{definition}

We note that $\MDP$ is $\NP$-hard to approximate to within any constant factor.
This was first established by Cheng and Wan~\cite{journals/tit/ChengW12}. (Earlier work showed $\NP$-hardness under randomized reductions~\cite{journals/tit/DumerMS03}, and subsequent work showed simplified deterministic reductions~\cite{journals/tit/AustrinK14,conf/coco/Micciancio14}.)

\begin{theorem}[\cite{journals/tit/ChengW12}] \label{thm:mdp-hardness}
For all constants $\gamma \geq 1$ and all prime powers $q$, $\gamma$-$\MDP_q$ is $\NP$-hard.
\end{theorem}

We now define a variant of the Nearest Codeword Problem with a stronger promise in the NO case.

\begin{definition} \label{def:gap-ncpvar}
For $\gamma = \gamma(k) \geq 1$, $\alpha = \alpha(k) > 0$, and a prime power $q$, the \emph{decisional $\gamma$-approximate Nearest Codeword Problem over $\F_q$ with distance promise $\alpha$} ($\gamma$-$\NCP_q^{\alpha}$) is the decision problem defined as follows. An instance consists of (a generator matrix $G \in \F_q^{n \times k}$ of) a code $\C$, a target vector $\vec{t} \in \F_q^n$, and a distance parameter $d \in \Z^+$. It is a:
\begin{itemize}
    \item YES instance if $\dist(\vec{t}, \C) \leq d$.
    \item NO instance if $\dist(\vec{t}, \C) > \gamma d$ and $d < \alpha \lambda_1(\C)$.
\end{itemize}
\end{definition}
We additionally define $\gamma$-$\NCP_q$ to be $\gamma$-$\NCP_q^{\infty}$. I.e., in $\gamma$-$\NCP_q$ there is no promised upper bound on $d$ in the NO case.
We note that $\gamma$-$\NCP_q^{\alpha}$ trivially reduces to $\gamma$-$\NCP_q^{\alpha'}$ for $\alpha' \geq \alpha$, and in particular $\gamma$-$\NCP_q^{\alpha}$ trivially reduces to ``plain'' $\gamma$-$\NCP_q$ for any $\alpha > 0$.

\subsection{Matrix Norms}
\label{sec:matrix-norms}

While there are many ways to define matrix norms, the most useful for us is the matrix $p \rightarrow q$ norm, where $p,q \in [1,\infty]$. We note in passing that matrix $p \to q$ norms over the real numbers have been studied from a computational standpoint; see, e.g.,~\cite{conf/soda/BhattiproluGGLT19}.

We can extend the definition of the matrix $p \to q$ norm to arbitrary (finite) fields $\F$ and to include $p = 0$ and $q = 0$, with $\norm{\vec{x}}_0$ being the Hamming weight as is standard. Note, however, that the matrix $0 \to 0$ ``norm'' is not in fact a matrix norm as it is not scale-invariant. However, it does satisfy the triangle inequality and it is sub-multiplicative. 
Ultimately, throughout this article, we abuse notation and nevertheless refer to it as a norm. Additionally, we employ a ``restricted'' version of the matrix $0 \rightarrow 0$ norm, where one only maximizes over those vectors $\vec{x} \in \C$ for some subspace $\C \subseteq \F^n$. More formally:
\begin{definition}
Let $n \in \Z^+$, let $\F$ be a field, let $T \in \GL_n(\F)$ be a matrix, and let $\C \subseteq \F^n$ be a subspace. The \emph{matrix $0 \rightarrow 0$ norm of $T$ restricted to the subspace $\C$} is the quantity
$$
\norm{T|_{\C}}_{0 \rightarrow 0} \defeq \max_{\vec{x} \in \C \setminus \{\vec{0}\}} \frac{\norm{T\vec{x}}_0}{\norm{\vec{x}}_0}.
$$
\end{definition} 
Note that the restricted matrix $0 \rightarrow 0$ norm is also sub-multiplicative.

\subsection{Distortion}
To compare pairs of general codes, we introduce a quantity that captures how similar they are. We formalize this with \emph{distortion}, similar to the formulation for the Lattice Distortion Problem \cite{conf/esa/BennettDS16} and distortion on metric embeddings.  We start by defining the distortion of a particular transformation between two codes.

\begin{definition} \label{def:matrix-dist}
Let $n, k \in \Z^+$, $k \leq n$, let $q$ be a prime power, let $\C_1, \C_2$ be $[n, k]_q$ codes, and let $T$ be a linear map such that $T(\C_1) = \C_2$. The \emph{distortion $\D_T(\C_1, \C_2)$ of $T$} the minimum value $D$ such that for all $\vec{x} \in \C_1$,
\[
D_1 \cdot \norm{\vec{x}}_0 \leq \norm{T \vec{x}}_0 \leq D_2 \cdot \norm{\vec{x}}_0
\]
for some values $D_1, D_2 > 0$ satisfying $D_2/D_1 \leq D$. Equivalently, 
$$
\D_T(\C_1, \C_2) :=
\Big( \max_{\vec{x} \in \C_1 \setminus \set{\vec{0}}} \frac{\norm{T\vec{x}}_0}{\norm{\vec{x}}_0} \Big) \big/ \Big(\min_{\vec{x} \in \C_1\setminus \set{\vec{0}}} \frac{\norm{T \vec{x}}_0}{\norm{\vec{x}}_0} \Big) \ \text{,}
$$
which is the definition we gave in \cref{sec:introduction}. Moreover, 
when $T \in \GL_n(\F_q)$ and thus is invertible, this is equivalent to
\begin{align} \label{eq:dist-via-matrix-norms}
\D_T(\C_1,\C_2) = \norm{T|_{\C_1}}_{0\rightarrow 0} \cdot \norm{T^{-1}|_{\C_2}}_{0\rightarrow 0}.
\end{align}
\end{definition}

Consequently, at least in the case when $T$ is invertible, the distortion between two codes $\C_1$ and $\C_2$ is inextricably tied to computing two matrix $0 \rightarrow 0$ norms, where one is restricted to $\C_1$ and the other is restricted to $\C_2$. Interestingly, as we show in \cref{lem:invertibilitylemma}, the assumption that $T$ is invertible can be made without loss of generality, so we may take \cref{eq:dist-via-matrix-norms} as an equivalent definition of distortion. 
Furthermore, while we require the codes in \cref{def:matrix-dist} to have the same block length $n$, the definition can be extended to codes with different block lengths by simply padding the shorter code with zeros (or, one could allow for non-square matrices $T$).

We can now define the distortion between two codes generically as the lowest possible distortion between them using any transformation.

\begin{definition} \label{def:code-dist-new}
Let $n, k \in \Z^+$, $k \leq n$ and let $q$ be a prime power. 
The \emph{distortion} $\D(\C_1, \C_2)$ between $[n, k]_q$ codes $\C_1$ and $\C_2$ is defined as the minimum value $D$ such that there exists $T \in \GL_n(\F_q)$ with $T(\C_1) = \C_2$ and $\D_T(\C_1,\C_2) = D$. In other words, 
\begin{equation}\label{eq:dist-equiv}
\D(\C_1, \C_2) \defeq \min_{\substack{T \in \GL_n(\F_q), \\ T(\C_1) = \C_2}}\D_T(\C_1, \C_2) \ \text{.}
\end{equation}

\end{definition}

\subsubsection{The Code Distortion Problem}
\label{sec:code-dist-prob-defs}

Using this definition of code distortion, we now define both the search and decision versions of the ($\gamma$-approximate) Code Distortion Problem (CDP).

\begin{definition} \label{def:searchcdp}
For $\gamma = \gamma(n) \geq 1$ and a prime power $q$, the \emph{$\gamma$-approximate Code Distortion Problem over $\F_q$} ($\gamma$-$\CDP_q$) is the search problem defined as follows. An instance consists of (generator matrices $G_1, G_2 \in \F_q^{m \times n}$ of) codes $\C_1, \C_2 \subseteq \F_q^n$, and the goal is to output $D$ such that $\D(\C_1, \C_2) \leq D \leq \gamma \cdot \D(\C_1, \C_2)$.
\end{definition}

\begin{definition} \label{def:cdp}
For $\gamma = \gamma(n) \geq 1$ and a prime power $q$, the \emph{decisional $\gamma$-approximate Code Distortion Problem over $\F_q$} ($\gamma$-$\GapCDP_q$) is the decision problem defined as follows. An instance consists of (generator matrices $G_1, G_2 \in \F_q^{m \times n}$ of) codes $\C_1, \C_2 \subseteq \F_q^n$ and $D \geq 1$. It is a:
\begin{itemize}
\item YES instance if $\D(\C_1, \C_2) \leq D$.
\item NO instance if $\D(\C_1, \C_2) > \gamma D$.
\end{itemize}
\end{definition}

The exact versions of these problems, $\CDP_q$ and $\GapCDP_q$, are then defined to be $1$-$\CDP_q$ and $1$-$\GapCDP_q$, respectively. Furthermore, the Linear Code Equivalence Problem is essentially the special case of $\CDP$ with $D = 1$. We formalize this relationship in \cref{prop:equiv-codes-scaling-up} in the next section.

\subsection{Basic Facts about Distortion}
\label{sec:prelims-dist-basic-facts}

We next present several basic facts about the distortion between codes. We start with the following bounds on $\D(\C_1, \C_2)$.

\begin{lemma} \label{lem:code-dist-basic-bounds}
Let $\C_1$ and $\C_2$ be $[n, k]_q$ codes. Then, for any linear map $T$ such that $T(\C_1) = \C_2$, $1 \leq \D_T(\C_1, \C_2) \leq n^2$. As a consequence, $1 \leq \D(\C_1, \C_2) \leq n^2$
\end{lemma}

\begin{proof}
For the lower bound, simply note that by the sub-multiplicativity of the matrix $0 \rightarrow 0$ norm on subspaces, it holds for all $T \in \GL_n(\F)$ such that $T(\C_1) = \C_2$ that
$$
1 = \norm{(T^{-1}T)|_{\C_1}}_{0 \rightarrow 0} \leq \norm{T^{-1}|_{\C_2}}_{0 \rightarrow 0} \cdot \norm{T|_{\C_1}}_{0 \rightarrow 0} = \D_T(C_1, \C_2).
$$
For the upper bound, note that for all $T \in \GL_n(\F)$ such that $T(\C_1) = \C_2$, it holds that
$$
\norm{T|_{\C_1}}_{0 \rightarrow 0} = \max_{\vec{x} \in \C_1 \backslash \{\vec{0}\}} \frac{\norm{T\vec{x}}_0}{\norm{\vec{x}}_0} \leq \max_{\vec{x} \in \C_1 \backslash \{\vec{0}\}} \norm{T\vec{x}}_0 \leq n \ \text{,}
$$
where the last inequality follows from the fact that $\C_2$ is an $[n,k]_q$ code and for all $x \in \C_1 \backslash \{\vec{0}\}$, $T\vec{x} \in \C_2 \backslash \{\vec{0}\}$. A similar argument establishes that $\norm{T^{-1}|_{\C_2}}_{0 \rightarrow 0} \leq n$, and consequently
$$
\D_T(\C_1, \C_2) = \norm{T|_{\C_1}}_{0\rightarrow 0} \cdot \norm{T^{-1}|_{\C_2}}_{0\rightarrow 0} \leq n^2 \ \text{,}
$$
as desired. The fact that $1 \leq \D(\C_1, \C_2) \leq n^2$ then follows from the definition of $\D(\C_1, \C_2)$.
\end{proof}

It is an open question if the upper bound in \cref{lem:code-dist-basic-bounds} can be reduced below $n^2$. Indeed, we do not know an example of a code with distortion $n^2$, but we do know an example of a code with distortion $n$, which we detail in the following.

\begin{lemma}
For every $n \in \Z^+$, there exist $[n, 2]_2$ codes $\C_1, \C_2$ such that $\D(\C_1, \C_2) = n$.
\end{lemma}

\begin{proof}
Let $\C_1,\C_2$ be the codes defined as follows:
\begin{align*}
    \C_1 = \left\{\vec{0}_n,\begin{pmatrix}
        1\\\vec{0}_{n-1}
    \end{pmatrix},\begin{pmatrix}
        0\\\vec{1}_{n-1}
    \end{pmatrix},\vec{1}_n\right\}, \quad 
    \C_2 = \left\{\vec{0}_n,\begin{pmatrix}
        1\\1\\0\\\vec{0}_{n-3}
    \end{pmatrix},\begin{pmatrix}
        0\\1\\1\\\vec{0}_{n-3}
    \end{pmatrix},\begin{pmatrix}
        1\\0\\1\\\vec{0}_{n-3}
    \end{pmatrix}\right\}.
\end{align*}
It is easy to check that $\C_1$ and $\C_2$ are linear codes.
Each non-zero codeword in $\C_2$ has Hamming weight $2$ (the constant weight of non-zero codewords in $\C_2$ is the only property of it that we are using) while the non-zero codewords in $\C_1$ have minimum weight $1$ and maximum weight $n$.  As such, the distortion for any transformation $T$ from $\C_1$ to $\C_2$ is
$$\max_{\vec{x} \in \C_1 \backslash \{\vec{0}\}} \frac{\norm{T\vec{x}}_0}{\norm{\vec{x}}_0}\max_{\vec{x} \in \C_1 \backslash \{\vec{0}\}} \frac{\norm{\vec{x}}_0}{\norm{T\vec{x}}_0}=\frac{2}{1} \cdot \frac{n}{2} = n \ \text{.} \qedhere $$
\end{proof}

Ultimately, while we are unable to prove a matching upper bound of $n$ on $\D(\C_1, \C_2)$, we are nevertheless able to prove such a bound when including \emph{non-linear transformations} $T$.

\begin{lemma} \label{lem:non-linear-dist-ub}
Let $\C_1$ and $\C_2$ be $[n,k]_q$ codes. Then, $\min_{T} \D_T(\C_1, \C_2) \leq n$, where the minimum is over all maps $T$ (including non-linear maps) such that $T(\C_1) = \C_2$.
\end{lemma}

\begin{proof}
List all non-zero codewords in $\C_1$ in monotonically increasing order of Hamming weight, and prepare a similar list for $\C_2$. Let $T$ be the map that takes codeword $i$ in the first list to codeword $i$ in the second list. Let $N := q^k - 1$, and let $1 \leq a_1 \leq \cdots \leq a_N \leq n$ and $1 \leq b_1 \leq \cdots \leq b_N \leq n$ be the Hamming weights of the non-zero codewords in the respective lists. Then, the distortion of $T$ is
$$
\D_T(\C_1, \C_2) = \left(\max_{i \in [N]}\frac{a_i}{b_i}\right) \cdot \left(\max_{j \in [N]}\frac{b_j}{a_j}\right) \ \text{.}
$$
Now put $i^* = \argmax_{i \in [N]} a_i / b_i$, and let $M = a_{i^*} / b_{i^*}$. We assume without loss of generality that $M \geq 1$ (otherwise, repeat the argument with the two codes flipped). Then, for all $j > i^*$, $b_j / a_j \leq n/a_{i^*} \leq n b_{i^*}/a_{i^*} = n/M$. And, for all $j \leq i^*$, $b_j / a_j \leq b_{i^*} / 1 = a_{i^*} / M \leq n / M$. Consequently, for all $j \in [N]$, $b_j / a_j \leq n/M$. Therefore, 
$$
\D_T(\C_1, \C_2) = \left(\max_{i \in [N]}\frac{a_i}{b_i}\right) \cdot \left(\max_{j \in [N]}\frac{b_j}{a_j}\right) \leq M \cdot \frac{n}{M} = n \ \text{,}
$$
as desired.
\end{proof}

We next show that if there exists a matrix $T$ such that $T \C_1 = \C_2$, then then exists an efficiently computable, \emph{invertible} matrix $T'$ such that $T' \C_1 = \C_2$ and $\D_T(\C_1, \C_2) = \D_{T'}(\C_1, \C_2)$.

\begin{lemma}
\label{lem:invertibilitylemma}
Let $T$ be a (not necessarily invertible) matrix such that $T(\C_1) = \C_2$ for $[n, k]_q$ codes $\C_1$, $\C_2$
such that $\D_T(\C_1, \C_2) \leq D$. Then there exists an invertible matrix $T'$ such that $T'(\C_1) = \C_2$ and $\D_{T'}(\C_1, \C_2)$.
\end{lemma}

\begin{proof}
Now consider a generator matrix $G_1 \in \F_q^{n \times k}$ of $\C_1$ and the corresponding generator matrix of $\C_2$, $G_2 := T(G_1) \in \F_q^{n \times k}$. 
Now, let $H_1, H_2 \in \F_q^{n \times (n-k)}$ be such that $A := (G_1 | H_1), B := (G_2 | H_2) \in \F_q^{n \times n}$ are full-rank.
(In particular, it suffices to take $H_1$ and $H_2$ to be (transposed) parity-check matrices for $\C_1$ and $\C_2$, respectively.)
Finally, let $T' := BA^{-1}$. 
Note that $T'$ is invertible, and that $T\vec{c} = T'\vec{c}$ for every $\vec{c} \in \C$. It follows that $T'(\C_1) = \C_2$ and $\D_{T'}(\C_1, \C_2) = \D_T(\C_1, \C_2) = D$, as needed.
\end{proof}

We next show that two codes $\C_1$ and $\C_2$ being linearly equivalent is equivalent to $\D(\C_1, \C_2) = 1$ up to ``scaling.''

\begin{proposition} \label{prop:equiv-codes-scaling-up}
    Let $\C_1, \C_2$ be $[n, k]_q$ codes for some $n, k \in \Z^+$ and prime power $q$.  Then $\D(\C_1, \C_2) = 1$ if and only if $\C_1 \otimes \vec{1}_{r_1}$ (padded with some number of zeros to reach block length $r_2n$) is linearly equivalent to $\C_2 \otimes \vec{1}_{r_2}$ for some $r_1,  r_2 \in \Z^+$, $r_1 \leq r_2$.
\end{proposition}
\begin{proof}
    We start with the backward direction.  Let $\C_1'$ be the result of padding $\C_1 \otimes \vec{1}_{r_1}$ with $r_2n-r_1n$ zeros and let $\C_2' = \C_2 \otimes \vec{1}_{r_2}$. Suppose we have $r_1, r_2 \in \Z^+$, $r_1 \le r_2$ such that $\C_1'$ is linearly equivalent to $\C_2'$.  Then, there must exist a 
    monomial matrix $T$ such that $T(\C_1') = \C_2'$ and $\norm{T(\vec{x} \otimes \vec{1}_{r_1})}_0 = \norm{\vec{x} \otimes \vec{1}_{r_1}}_0$ for all $\vec{x} \in \C_1$.
    Additionally, we can define $S_1 = I_n\otimes \vec{1}_{r_1}$, which maps $\C_1$ to $\C_1'$ with uniform scaling $\norm{S_1\vec{x}}_0 = r_1\norm{\vec{x}}_0$ for any $\vec{x} \in \C_1$, and $S_2 = I_n\otimes \vec{e}_{1}^{\top}$,\footnote{Here $\vec{e}_{1} \in \F_q^{r_2}$ is the first standard normal basis vector in $r_2$ dimensions.} which maps $\C_2'$ to $\C_2$ with uniform scaling of $\norm{S_2\vec{x}}_0 = \frac{1}{r_2}\norm{\vec{x}}_0$ for any $\vec{x} \in \C_2'$.
    We now compose these to define a transformation $T' = S_2TS_1$ that maps $\C_1$ to $\C_2$ and note that for any $\vec{x} \in \C_1$,
    $$\frac{\norm{T'\vec{x}}_0}{\norm{\vec{x}}_0} 
    = \frac{\norm{S_2TS_1\vec{x}}_0}{\norm{\vec{x}}_0}
    = \frac{1}{r_2} \cdot \frac{\norm{TS_1\vec{x}}_0}{\norm{\vec{x}}_0} 
    = \frac{1}{r_2} \cdot \frac{\norm{S_1\vec{x}}_0}{\norm{\vec{x}}_0} 
    = \frac{r_1}{r_2} \cdot \frac{\norm{\vec{x}}_0}{\norm{\vec{x}}_0} = \frac{r_1}{r_2}\text{.}$$
    As this is true for all $\vec{x} \in \C_1$, we can see that 
    $\zznorm{T'|_{\C_1}} =\frac{r_1}{r_2}$ and $ \zznorm{(T')^{-1}|_{\C_2}} = \frac{r_2}{r_1}$. Thus, $\D(\C_1,\C_2) = \D_T(\C_1,\C_2) = 1$.

    For the forward direction, we are now given that $\D(\C_1,\C_2) = 1$, i.e., there exists some $T$ such that $T(\C_1) = \C_2$ and $\norm{T\vec{x}}_0 = c\norm{\vec{x}}_0$ for all $\vec{x} \in \C_1$ and some $c \in \Q^+$ (we assume without loss of generality that $c \le 1$).  Let $c = \frac{r_1}{r_2}$ for some $r_1,r_2 \in \Z^+$, and once again take $\C_1' = \C_1\otimes\vec{1}_{r_1}$ (padded with $r_2n-r_1n$ zeros) and $\C_2' = \C_2\otimes\vec{1}_{r_2}$.  Finally, we %
    redefine $S_1 = \left(I_n\otimes\begin{pmatrix}
        1 & \vec{0}_{r_1-1}
    \end{pmatrix}\right)$, $S_2 = I_n\otimes \vec{1}_{r_2}$.  As before, note that $S_1(\C_1') = \C_1$, $S_2(\C_2) = \C_2'$, $\norm{S_1\vec{x}}_0 = \frac{1}{r_1}\norm{\vec{x}}_0$ for all $\vec{x} \in \C_1'$, and $\norm{S_2\vec{x}}_0 = r_2\norm{\vec{x}}_0$ for all $\vec{x} \in \C_2$.  We define $T' = S_2TS_1$ so that $T'$ maps $\C_1'$ to $\C_2'$, and again compute that 
    $$\norm{T'\vec{x}}_0 = \norm{S_2TS_1\vec{x}}_0 = r_2\norm{TS_1\vec{x}}_0 = cr_2\norm{S_1\vec{x}}_0 = \frac{cr_2}{r_1}\norm{\vec{x}}_0 = \norm{\vec{x}}_0 \ \text{.}$$
    This is true for all $\vec{x} \in \C_1'$, so $\C_1' = \C_1\otimes\vec{1}_{r_1}$ is equivalent to $\C_2' = \C_2\otimes\vec{1}_{r_2}$ as required to complete the proof.
\end{proof}

\subsection{Reductions}
\label{sec:reductions}

In this paper, we employ three different types of reductions. The first and most general notion is that of a \emph{Turing reduction}. Formally, there is a Turing reduction from a decision problem $A$ to a decision problem $B$ if and only if there exists an algorithm for $A$, given oracle access to $B$. In other words, $A$ Turing-reduces to $B$ if, given a subroutine for $B$, one can decide $A$. The second type of reduction is a \emph{Cook reduction}. Formally, a Cook reduction is a polynomial-time Turing reduction, i.e., $A$ Cook-reduces to $B$ if and only if $A$ Turing-reduces to $B$ and the overall algorithm runs in polynomial time. The third and final type of reduction is that of a \emph{Karp reduction} (a.k.a., a \emph{polynomial-time many-one reduction}). Formally, a Karp reduction is a deterministic and polynomial-time algorithm that maps YES and NO instances of $A$ to YES and NO instances of $B$, respectively.

\section{The Complexity of the Code Distortion Problem}
\label{sec:complexity}

We now study some basic facts about the complexity of both the exact and approximate versions of $\GapCDP$. Given the relationship of $\GapCDP$ to the matrix $0 \rightarrow 0$ norm (\cref{sec:matrix-norms}), we start by discussing the complexity of computing both the unrestricted and restricted versions of the matrix $0 \rightarrow 0$ norm.

\subsection{The Complexity of the Matrix \texorpdfstring{$0\rightarrow 0$}{0 to 0} Norm}

Interestingly, the (unrestricted) matrix $0 \rightarrow 0$ norm is easy to compute.

\begin{proposition} \label{prop:zznorm-val}
Let $\F$ be a field, let $m, n \in \Z^+$, and let $T = (\vec{t}_1, \ldots, \vec{t}_n) \in \F^{m \times n}$. Then,
\[
\norm{T}_{0 \rightarrow 0} = \max_{i \in [n]} \norm{\vec{t}_i}_0 \ \text{.}
\]
\end{proposition}

\begin{proof}
By the triangle inequality, it holds for all $\vec{x} \in \F^n \setminus \set{\vec{0}}$ that
\fullornot{
\[
\frac{\norm{T\vec{x}}_0}{\norm{\vec{x}}_0}
= \frac{\norm{\sum_{i \in \supp(\vec{x})} x_i \vec{t}_i}_0}{\card{\supp(\vec{x})}}
\leq \frac{\sum_{i \in \supp(\vec{x})} \norm{x_i \vec{t}_i}_0}{\card{\supp(\vec{x})}}
\leq \frac{\card{\supp(\vec{x})} \cdot \max_{i \in \supp(\vec{x})} \norm{\vec{t}_i}_0}{\card{\supp(\vec{x})}}
= \max_{i \in \supp(\vec{x})} \norm{\vec{t}_i}_0 \ \text{.}
\]
}{
\begin{align*}
\frac{\norm{T\vec{x}}_0}{\norm{\vec{x}}_0}
&= \frac{\norm{\sum_{i \in \supp(\vec{x})} x_i \vec{t}_i}_0}{\card{\supp(\vec{x})}} \\
&\leq \frac{\sum_{i \in \supp(\vec{x})} \norm{x_i \vec{t}_i}_0}{\card{\supp(\vec{x})}} \\
& \leq \frac{\card{\supp(\vec{x})} \cdot \max_{i \in \supp(\vec{x})} \norm{\vec{t}_i}_0}{\card{\supp(\vec{x})}} \\
&= \max_{i \in \supp(\vec{x})} \norm{\vec{t}_i}_0 \ \text{.}
\end{align*}
}
Furthermore,
$\max_{i \in \supp(x)} \norm{\vec{t}_i}_0 \leq \max_{i \in [n]} \norm{\vec{t}_i}_0$, so $\norm{T}_{0 \rightarrow 0} \leq \max_{i \in [n]} \norm{\vec{t}_i}_0$.
On the other hand,
\[
\norm{T}_{0 \rightarrow 0}
\geq \max_{i \in [n]} \frac{\norm{T\vec{e}_i}_0}{\norm{\vec{e}_i}_0} = \max_{i \in [n]} \norm{\vec{t}_i}_0 \ \text{.} \qedhere
\]
\end{proof}

We note that an immediate consequence of \cref{prop:zznorm-val} is that there is a $\poly(m, n)$-time algorithm for computing $\norm{T}_{0 \to 0}$ for a matrix $T \in \F^{m \times n}$: simply compute the maximum Hamming weight of a column of $T$.
That said, we show that matrix $0\rightarrow 0$ norm \emph{restricted to a subspace} is hard to compute, even approximately. We formalize this by defining a decisional version of the problem.

\begin{definition}
Let $\F$ be a field, let $m,n,k,q \in \Z^+$, and let $\gamma = \gamma(m,n) \geq 1$. The \emph{decisional $\gamma$-approximate matrix $0 \rightarrow 0$ norm problem over a subspace of $\F_q$} ($\gamma$-GapNorm$^{0 \rightarrow 0}_{q}$) is the promise problem defined as follows. An instance consists of a basis $G \subseteq \F_q^{n \times k}$ of a subspace $\C \subseteq \F_q^n$, a matrix $T \in \F_q^{m \times n}$, and threshold $r \geq 0$. It is a:
\begin{itemize}
\item YES instance if $\norm{T|_\C}_{0 \rightarrow 0} \leq r$.
\item NO instance if $\norm{T|_\C}_{0 \rightarrow 0} > \gamma r$.
\end{itemize}
\end{definition}

We will now give a reduction from $\gamma'$-$\MDP_2$ to $\gamma$-$\GapNorm^{0 \rightarrow 0}_2$ for which it suffices to take $\gamma' = O(\gamma)$. This result is due to Golovnev and Stephens-Davidowitz~\cite{golovnev-stephens-davidowitz-personal-comm-26}.

\begin{theorem}[{\cite{golovnev-stephens-davidowitz-personal-comm-26}}]
\label{thm:MDPtoGapNormReduction}
Let $\gamma \geq 1$ and let $\epsilon > 0$ be a constant. Then, for all $\gamma' > \big(\frac{1 + \epsilon}{1 - \epsilon}\big) \cdot \gamma$, there is a Karp reduction from $\gamma'$-$\MDP_{2}$ to $\gamma$-$\GapNorm_2^{0\rightarrow 0}$.
\end{theorem}

\begin{proof}
Let $(G \in \F_2^{m \times n}, r')$ be an instance of $\gamma'$-$\MDP_2$, where $m,n \in \Z^+$. Moreover, let $T$ be a generator matrix of an $[m,n]_2$ code that is $\eps$-balanced, i.e., all non-zero codewords in $\C(T)$ have Hamming weight in the interval $[(1-\eps) \frac{m}{2}, (1 + \eps)\frac{m}{2}]$. By a result of Ta-Shma \cite{Ta-Shma2017}, such a $T$ exists and, moreover, can be constructed in deterministic polynomial time.

Now, if $(G, r')$ is a YES instance of $\gamma'$-MDP$_2$, then for all $\vec{x} \in \C \backslash \{\vec{0}\}$ with $\norm{\vec{x}}_0 = \lambda_1(\C)$,
$$
\norm{T|_{\C}}_{0 \rightarrow 0} \geq \frac{\norm{T\vec{x}}_0}{\norm{\vec{x}}_0} \geq \frac{\norm{T\vec{x}}_0}{r'} \geq \frac{(1 - \epsilon)m/2}{r'}.
$$
Here, the first inequality uses the definition of the restricted matrix $0 \rightarrow 0$ norm, the second uses the fact that $(G,r')$ is a YES instance, and the third uses the fact that $T$ generates an $\eps$-balanced code. 

On the other hand, if $(G, r')$ is a NO instance of $\gamma'$-MDP$_2$, then for all $\vec{x} \in \C \backslash \{\vec{0}\}$,
$$
\frac{\norm{T\vec{x}}_0}{\norm{\vec{x}}_0} < \frac{(1 + \epsilon)m/2}{\gamma' r'},
$$
which implies
$$
\norm{T|_{\C}}_{0 \rightarrow 0} = \max_{\vec{x} \in \C \backslash \{\vec{0}\}}\frac{\norm{T\vec{x}}_0}{\norm{\vec{x}}_0} < \frac{(1 + \epsilon)m/2}{\gamma' r'}.
$$

Altogether, with $\gamma = \frac{1-\eps}{1 + \eps} (\gamma' + \delta)$ and $r = \frac{(1 + \epsilon) m/2}{(\gamma' + \delta) r'}$ for any $\delta > 0$, it holds that $\norm{T|_\C}_{0 \rightarrow 0} > \gamma r$ (a YES instance of co-$\gamma$-$\GapNorm_2^{0\rightarrow0}$) implies a YES instance of $\gamma'$-MDP$_2$, and $\norm{T|_\C}_{0 \rightarrow 0} \leq r$ (a NO instance of co-$\gamma$-$\GapNorm_2^{0\rightarrow0}$) implies a NO instance of $\gamma'$-MDP$_2$.
\end{proof}

Consequently, by \cref{thm:mdp-hardness}, $\gamma$-$\GapNorm^{0\rightarrow0}_2$ is $\coNP$-hard. In fact, it is $\coNP$-complete.

\begin{corollary} \label{cor:zznorm-conp-hard}
For all $\gamma \geq 1$, $\gamma$-$\GapNorm_2^{0\rightarrow0}$ is $\coNP$-complete.
\end{corollary}
\begin{proof}
Since $\gamma$-$\GapNorm^{0\rightarrow0}_2$ is $\coNP$-hard, it suffices to prove that $\gamma$-$\GapNorm_2^{0\rightarrow0} \in \coNP$. But this is plain, since if $(G,T,r)$ is a YES instance, then for all $\vec{x} \in \C \backslash \{\vec{0}\}$, $\|T \vec{x}\|_0 / \|\vec{x}\|_0 \leq r$, and if $(G,T,r)$ is a NO instance, then there exists $\vec{x} \in \C \backslash \{\vec{0}\}$ for which $\|T \vec{x}\|_0 / \|\vec{x}\|_0 \geq \gamma r$. 
\end{proof}

\subsection{GapCDP is in \texorpdfstring{$\Sigma_2^\P$}{Sigma\_2\^P}}

We now show that the code distortion problem is in $\Sigma_2^\P$, the second level of the polynomial hierarchy. Recall that a language $L \in \Sigma_2^\P$ if and only if there exists two polynomials $p_1$ and $p_2$ as well as a polynomial-time deterministic Turing machine $M$ such that for all inputs $x \in \{0,1\}^*$, $x \in L$ if and only if $(\exists y \in \{0,1\}^{p_1(|x|)})(\forall z \in \{0,1\}^{p_2(|x|)}) : M(x,y,z) = 1$.

\begin{theorem} \label{thm:sigma2-containment}
For all $q \in \Z^+$, $\GapCDP_q \in \Sigma_2^\P$.
\end{theorem}

\begin{proof}
Let $(G_1, G_2, D)$ be an instance of $\GapCDP_q$, where $G_1$ and $G_2$ generate the codes $\C_1$ and $\C_2$, respectively. By definition, $(G_1, G_2, D)$ is a YES instance if and only if $\D(\C_1, \C_2) \leq D$. This holds if and only if there exists $T \in \GL_m(\F_q)$ with $T(\C_1) = \C_2$ such that for all $\vec{x}, \vec{y} \in \C_1 \backslash \{\vec{0}\}$, 
$$
\left(\frac{\norm{T\vec{x}}_0}{\norm{\vec{x}}_0}\right) \cdot \left(\frac{\norm{T\vec{y}}_0}{\norm{\vec{y}}_0}\right)^{-1} \leq D.
$$
For all such $T$, $\vec{x}$, and $\vec{y}$, the above inequality is checkable in deterministic polynomial time. Moreover, the size of all such $T$, $\vec{x}$, and $\vec{y}$ is a polynomial in the size of the input $(G_1, G_2, D)$. Consequently, as the order of the $\exists$ and $\forall$ quantifiers in the above reformulation of $\GapCDP_q$ is consistent with the class $\Sigma_2^\P$, it holds that $\GapCDP_q \in \Sigma_2^\P$, as desired.
\end{proof}

As mentioned in \cref{sec:open-questions}, we suspect $\GapCDP_q$ is also $\Sigma_2^\P$-hard, thus making it $\Sigma_2^\P$-complete. However, proving this remains an interesting open question.

\subsection{GMSS for Codes}

We will now give a reduction from $\gamma$-$\MDP$ to $\gamma$-$\NCP$ that is analogous to the seminal reduction from the Shortest Vector Problem to the Closest Vector Problem on lattices due to~\cite{journals/ipl/GoldreichMSS99}. We will use this result in our proof of $\NP$-hardness of approximate $\CDP$.

\begin{theorem}[GMSS for Codes] \label{thm:gmss-for-codes}
Let $m, n \in \Z^+$ with $n \leq m$, let $\gamma = \gamma(m, n) \geq 1$, and let $q$ be a prime power. Then there is a $\poly(m, q)$-time Turing reduction from $\gamma$-$\MDP_q$ on $[m, n]_q$ codes to $\gamma$-$\NCP_q^{1/\gamma}$ (and hence also $\gamma$-$\NCP_q$) on $[m, n - 1]_q$ codes.
\end{theorem}

\begin{proof}
Let $(G = (\vec{g}_1, \ldots, \vec{g}_n) \in \F_q^{m \times n}, d)$ be the input instance of $\gamma$-$\MDP_q$.
The reduction does the following. It constructs $(q-1) n$ many instances $(G_i, \vec{t}_{i,j} := -j \cdot \vec{g}_i, d)$ of $\NCP_q$ for $1 \leq i \leq n$ and $j \in \F_q^*$, where $G_i := (\vec{g}_1, \ldots, \vec{g}_{i-1}, \vec{g}_{i+1}, \ldots, \vec{g}_n)$ is ``$G$ with its $i$th column removed.''
It then calls its $\gamma$-$\NCP_q^{1/\gamma}$ oracle on each instance $(G_i, \vec{t}_{i,j}, d)$. The reduction outputs YES if the oracle responds with YES on some input, and otherwise it outputs NO.

It is clear that the reduction runs in the stated amount of time, and it remains to show its correctness. 
Suppose that the input is a YES instance. Then there exists $\vec{a} \in \F_q^n \setminus \set{\vec{0}}$ such that $\norm{G \vec{a}}_0 \leq d$. Let $a_i$ be a non-zero coordinate of $\vec{a}$, and let $\vec{a}' := (a_1, \ldots, a_{i-1}, a_{i+1}, \ldots, a_n) \in \F_q^{n-1}$ be ``$\vec{a}$ with its $i$th coordinate deleted.'' Then for $j = a_i$, we have that
\[
\norm{G_i \vec{a}' - \vec{t}_{i,j}}_0
= \norm{G_i \vec{a}' + j \vec{g}_i}_0
= \big\| \sum_{\ell \neq i} a_{\ell} \vec{g}_{\ell} + a_i \vec{g}_i \big\|_0
= \norm{G \vec{a}}_0 \leq d \ \text{.}
\]
So, the $\gamma$-$\NCP_q^{1/\gamma}$ oracle outputs YES on input $(G_i, \vec{t}_{i,j}, d)$, as needed.

Now, suppose that the input is a NO instance.
Let $\C := \C(G)$ and let $\C_i := \C(G_i)$.
We note that for every $1 \leq i \leq n$ and $j \in \F_q^*$,
\[
\C_i \subset \C \text{ and } \C_i - \vec{t}_{i,j} \subset \C \setminus \set{\vec{0}} \ \text{.}
\]
It follows that $\lambda_1(\C_i) \geq \lambda_1(\C)$ and that $\dist(\vec{t}_{i,j}, \C_i) \geq \lambda_1(\C)$ for all $1 \leq i \leq n$ and $j \in \F_q^*$. Because the input is a NO instance, $\lambda_1(\C) > \gamma d$, and so we have that for every $1 \leq i \leq n$, $\lambda_1(\C_i) \geq \lambda_1(\C) > \gamma d$, and therefore $d < \lambda_1(\C_i)/\gamma$. Moreover, for every $1 \leq i \leq n$ and $j \in \F_q^*$, $\dist(\vec{t}_{i,j}, \C_i) \geq \lambda_1(\C) > \gamma d$.
Therefore, each instance $(G_i, \vec{t}_{i,j}, d)$ is a NO instance of $\gamma$-$\NCP_q^{1/\gamma}$, as needed.
\end{proof}

\subsection{\texorpdfstring{$\NP$}{NP}-hardness of the Code Distortion Problem}
\label{sec:np-hardness-cdp}

We next give a reduction from $\gamma'$-$\NCP_q^{\alpha}$ to $\gamma$-$\CDP_q$ for which it suffices to take $\gamma' = O(\gamma)$. 

\begin{theorem} \label{thm:ncp-to-cdp}
Let $\gamma \geq 1$, let $q$ be a prime power, and let $D > 1$ be a constant. Then for any $\gamma' > \ceil{\frac{D+1}{D-1}}D \cdot \gamma$, there is a Karp reduction from $\gamma'$-$\NCP_{q}^{1/\gamma'}$ to $\gamma$-$\CDP$ with distortion $D$.
\end{theorem}

\begin{proof}
Let $(G \in \F_q^{n \times k}, \vec{t} \in \F_q^n, d \in \Z^+)$ be an instance of $\gamma'$-$\NCP_{q}^{1/\gamma'}$. Define
\[
G_1 := \begin{pmatrix}
G & 0 & \vec{0} \\
0 & G & \vec{0} \\
0 & 0 & \vec{1}_r 
\end{pmatrix} \ \text{,} \qquad 
G_2 := \begin{pmatrix}
G & 0 & -\vec{\vec{t}} \\
0 & G & \vec{0} \\
0 & 0 & \vec{1}_r 
\end{pmatrix}
\]
with
$r := \ceil{\frac{D+1}{D-1}} \cdot d$. The reduction outputs the $\gamma$-$\CDP_q$ instance $(G_1, G_2, D)$.

It is clear that the reduction is efficient, and it remains to shows its correctness.
Let $\C := \C(G)$, $\C_1 := \C(G_1)$, and $\C_2 := \C(G_2)$.
Suppose that the input is a YES instance of $\gamma'$-$\NCP_{q}^{1/\gamma'}$.
Let $\vec{t}' \in \C - \vec{t}$ be such that $\norm{\vec{t}'}_0 = \dist(\vec{t}, \C)$, i.e., $\vec{t}'$ is a shortest codeword in the coset $\C - \vec{t}$.
Define
\[
G_2' := \begin{pmatrix}
G & 0 & \vec{\vec{t}}' \\
0 & G & \vec{0} \\
0 & 0 & \vec{1}_r 
\end{pmatrix}
\]
Notice that $\C(G_2) = \C(G_2')$, and, because the input is a YES instance, that $\norm{\vec{t}'}_0 \leq d$.%
\footnote{Actually computing $\vec{t}'$ amounts to solving $\gamma'$-$\NCP_{q}^{1/\gamma'}$, which is a hard problem. However, here we are only using $\vec{t}'$ and $G_2'$ for analysis, and we are not computing them.}
Let $G_1 = (\vec{g}_1, \ldots, \vec{g}_{2n+1})$, let $G_2' = (\vec{g}_1', \ldots, \vec{g}_{2n+1}')$, and let $M$ be an invertible linear map such that $M : \vec{g}_i \mapsto \vec{g}_i'$ for every $i$. Such a full-rank map $M$ exists because $G_1$ and $G_2'$ are full-rank, and it is clear that $M$ maps $\C_1$ to $\C_2$.

Let $\vec{x} \in \C_1$ be an arbitrary non-zero codeword.
We can write $\vec{x} = (\vec{x}', \vec{x}'', a \vec{1}_r)$ for some $\vec{x}', \vec{x}'' \in \C$ and $a \in \F_q$.
If $a = 0$, then $\norm{M\vec{x}}_0 = \norm{\vec{x}}_0$, and so $\norm{M\vec{x}}_0/\norm{\vec{x}}_0 = 1$.
On the other hand, assume that $a \neq 0$. Then $M\vec{x} = \vec{x} + (a \vec{t}', \vec{0})$,
and so by triangle inequality and the fact that $\norm{\vec{t}'}_0 \leq d$, $\norm{M\vec{x}}_0 \leq \norm{\vec{x}}_0 + d$.
Therefore,
\begin{equation} \label{eq:dist-ub}
\frac{\norm{M\vec{x}}_0}{\norm{\vec{x}}_0} \leq \frac{\norm{\vec{x}}_0 + d}{\norm{\vec{x}}_0} \leq
\frac{\norm{\vec{x}}_0 + d}{\norm{a \vec{1}_r}_0} =
1 + d/r \ \text{.}
\end{equation}
Furthermore, $\vec{x} = M\vec{x} - (a\vec{t}', \vec{0})$, and so $\norm{\vec{x}}_0 = \norm{M\vec{x} - (a\vec{t}', \vec{0})}_0 \leq \norm{M\vec{x}}_0 + \norm{\vec{t}'}_0 \leq \norm{M\vec{x}}_0 + d$. It follows that $\norm{M\vec{x}}_0 \geq \norm{\vec{x}}_0 - d$, and therefore,
\begin{equation} \label{eq:dist-lb}
\frac{\norm{M\vec{x}}_0}{\norm{\vec{x}}_0} 
\geq \frac{\norm{\vec{x}}_0 - d}{\norm{\vec{x}}_0} 
= 1 - \frac{d}{\norm{\vec{x}}_0}
\geq 1 - \frac{d}{\norm{a \vec{1}_r}_0} =
1 - d/r \ \text{.}
\end{equation}

Recalling that $r  = \ceil{\frac{D+1}{D-1}}d \ge \frac{D+1}{D-1}d$, we get that $d/r \le \frac{D-1}{D+1}$. By combining \cref{eq:dist-ub,eq:dist-lb} we then have that
\begin{align*}
    \D(\C_1, \C_2) &\le \max_{\vec{x} \in \C_1 \setminus \set{\vec{0}}} \Big(\frac{\norm{M\vec{x}}_0}{\norm{\vec{x}}_0}\Big) \big/ \min_{\vec{x} \in \C_1 \setminus \set{\vec{0}}} \Big(\frac{\norm{M\vec{x}}_0}{\norm{\vec{x}}_0}\Big)  \\
    &\leq \left(1 + d/r\right)/\left(1 - d/r\right)\\
    &\le \left(1 + \frac{D-1}{D+1}\right)/\left(1 - \frac{D-1}{D+1}\right)\\
    &= \left(\frac{2D}{D+1}\right)/\left(\frac{2}{D+1}\right)\\
    &= D,
\end{align*}
as needed.

Now, suppose that the input is a NO instance, and let $M$ be an invertible linear map such that $M(\C_1) = \C_2$.
Let $\vec{y} := M \vec{g}_{2n+1} = M \cdot (\vec{0}, \vec{1}_r)$, and note that $\vec{y} = (\vec{y}' - a \vec{t}, \vec{y}'', a \vec{1}_r)$ for some $\vec{y}', \vec{y}'' \in \C$ and $a \in \F_q$.
If $a \neq 0$, then
\[
\norm{\vec{y}}_0 = \norm{\vec{y}' - a \vec{t}}_0 + \norm{\vec{y}''}_0 + \norm{a \vec{1}_r}_0 \geq \dist(-a \vec{t}, \C) + r > \gamma' d + r \ \text{,}
\]
where we have used that $\dist(-a \vec{t}, \C) = \dist(\vec{t}, \C)$ for $a \neq 0$.
On the other hand, if $a = 0$, then
\[
\norm{\vec{y}}_0 = \norm{\vec{y}'}_0 + \norm{\vec{y}''}_0 \geq \lambda_1(\C) > \gamma' d \ \text{,}
\]
where the first inequality holds because at least one of $\vec{y}'$ and $\vec{y}''$ is a non-zero codeword in $\C$ (since $\vec{y} \neq \vec{0}$), and the second inequality holds because of the assumption that $d < \lambda_1(\C)/\gamma'$, which follows from the definition of NO instances of $\gamma'$-$\NCP_q^{1/\gamma'}$.
In either case, we have that
\begin{equation} \label{eq:C1-to-C2}
\max_{\vec{x} \in \C_1 \setminus \set{\vec{0}}} \frac{\norm{M\vec{x}}_0}{\norm{\vec{x}}_0}
\geq \frac{\norm{M \vec{g}_{2n+1}}_0}{\norm{\vec{g}_{2n+1}}_0}
= \frac{\norm{\vec{y}}_0}{r}
> \gamma' d/r \ \text{.}
\end{equation}

Now, let $\vec{x}' \in \C$ be such that $\norm{\vec{x}'}_0 = \lambda_1(\C)$, and define $\vec{x}_1 := (\vec{x}', \vec{0}, \vec{0}), \vec{x}_2 := (\vec{0}, \vec{x}', \vec{0})$. Note that $\vec{x}_1, \vec{x}_2 \in \C_2$, and that $\vec{y}_1 := M^{-1} \vec{x}_1 = (\vec{y}_1', \vec{y}_1'', a_1 \vec{1}_r)$ and $\vec{y}_2 := M^{-1} \vec{x}_2 = (\vec{y}_2', \vec{y}_2'', a_2 \vec{1}_r)$ for some $\vec{y}_1', \vec{y}_1'', \vec{y}_2', \vec{y}_2'' \in \C$ and $a_1, a_2 \in \F_q$.
Furthermore, because $\vec{x}_1$ and $\vec{x}_2$ are linearly independent and $M$ is invertible, $\vec{y}_1$ and $\vec{y}_2$ must also be linearly independent. So, at least one of $\vec{y}_1', \vec{y}_1'', \vec{y}_2', \vec{y}_2''$ is non-zero, and therefore $\max \set{\norm{M^{-1} \vec{x}_1}_0, \norm{M^{-1} \vec{x}_2}_0} \geq \lambda_1(\C)$.
It follows that
\fullornot{
\begin{equation} \label{eq:C2-to-C1}
\max_{\vec{x} \in \C_2 \setminus \set{\vec{0}}} \frac{\norm{M^{-1}\vec{x}}_0}{\norm{\vec{x}}_0}
\geq \max \Big\{ \frac{\norm{M^{-1} \vec{x}_1}_0}{\norm{\vec{x}_1}_0}, \frac{\norm{M^{-1} \vec{x}_2}_0}{\norm{\vec{x}_2}_0} \Big\}
= \max \set{\norm{M^{-1} \vec{x}_1}_0, \norm{M^{-1} \vec{x}_2}_0}/\lambda_1(\C) \geq 1 \ \text{.}
\end{equation}
}{
\begin{align} \label{eq:C2-to-C1}
\begin{split}
\max_{\vec{x} \in \C_2 \setminus \set{\vec{0}}} \frac{\norm{M^{-1}\vec{x}}_0}{\norm{\vec{x}}_0}
&\geq \max \Big\{ \frac{\norm{M^{-1} \vec{x}_1}_0}{\norm{\vec{x}_1}_0}, \frac{\norm{M^{-1} \vec{x}_2}_0}{\norm{\vec{x}_2}_0} \Big\} \\
&= \max \set{\norm{M^{-1} \vec{x}_1}_0, \norm{M^{-1} \vec{x}_2}_0}/\lambda_1(\C) \geq 1 \ \text{.}
\end{split}
\end{align}
}
Combining \cref{eq:C1-to-C2,eq:C2-to-C1} then gives
\[
\D(\C_1, \C_2) > \gamma' d/r = \frac{\gamma'}{\ceil{\frac{D+1}{D-1}}} > \frac{\ceil{\frac{D+1}{D-1}} \cdot \gamma D}{\ceil{\frac{D+1}{D-1}}} = \gamma D \ \text{,}
\]
as needed.
\end{proof}

From \cref{thm:gmss-for-codes,thm:ncp-to-cdp} and the fact that $\gamma$-$\MDP$ is $\NP$-hard for any constant $\gamma \geq 1$, we conclude that the $\CDP$ is $\NP$-hard to approximate to within any constant factor. 

\hardnessapprox*

\begin{proof}
Combine the $\NP$-hardness result in \cref{thm:mdp-hardness} with the reductions in \cref{thm:gmss-for-codes,thm:ncp-to-cdp}.
\end{proof}

We note that restriction on the field size $q$ in \cref{thm:intro-hardness-of-approx} is due the running time of the reduction in \cref{thm:gmss-for-codes}.

\section{Algorithms for Code Distortion}
\label{sec:algorithms}

We now turn to giving algorithms for code distortion.

\subsection{An Exact Algorithm}
\label{sec:exact-alg}

We first analyze the running time of a brute force exact algorithm for $\CDP$, which is the best exact algorithm we know.

\begin{lemma} \label{lem:brute-force-alg}
Let $n, k \in \Z^+$, $k \leq n$, let $q$ be a prime power, and let $\C_1, \C_2$ be $[n, k]_q$ codes.  The distortion $\D(\C_1,\C_2)$ can be computed in 
$O^*(q^{k(k+1)})$ time.
\end{lemma}

\begin{proof}
A linear transformation $T$ with $T \C_1 = \C_2$ must map a generator matrix of $\C_1$ to a generator matrix of $\C_2$. So, it suffices to fix a generator matrix of $G_1 \in \F_q^{n \times k}$, enumerate generator matrices $G_2 \in \F_q^{n \times k}$ of $\C_2$, and compute $\D_T(\C_1, \C_2)$ for the induced map $T$ such that $T G_1 = G_2$.

To enumerate generator matrices $G_2$ of $\C_2$, it suffices to compute $G_2 := G_2' U$ for a fixed generator matrix $G_2'$ of $\C_2$ and each $U \in \GL_k(\F_q)$.
In order to enumerate elements of $\GL_k(\F_q)$, it suffices to enumerate all matrices $U \in \F_q^{k \times k}$, and discard them if they are singular.
This takes $O^*(q^{k^2})$ time.

To compute $\D_T(\C_1, \C_2)$ for a given map $T$, it suffices to compute $\norm{T\vec{x}}_0/\norm{\vec{x}}_0$ for all $\vec{x} \in \C_1 \setminus \set{\vec{0}}$ and $\norm{T^{-1} \vec{y}}_0/\norm{\vec{y}}_0$ for all $\vec{y} \in \C_2 \setminus \set{\vec{0}}$. This can be done in $O^*(q^k)$ time since $\card{\C_1} = \card{\C_2} = q^k$.
So, overall the algorithm runs in $O^*(q^{k(k+1)})$, as needed.
\end{proof}

We then get the following corollary.

\begin{corollary} \label{cor:poly-alg}
Let $n, k \in \Z^+$, $k \leq n$, let $q$ be a prime power, and let $\C_1, \C_2$ be $[n, k]_q$ codes.  The distortion $\D(\C_1,\C_2)$ can be computed in $\poly(n)$ time if $k = O\left(\sqrt{\log_q n}\right)$.
\end{corollary}

Notably, \cref{lem:brute-force-alg} runs in roughly $q^{k^2}$ time, and at a minimum it would be desirable to find a single-exponential, roughly $q^k$-time algorithm. Although we do not give an exact such $q^k$-time algorithm, we do give a roughly $q^k$-time \emph{approximation} algorithm for $\CDP_q$.

\subsection{A General Approximation Algorithm}
\label{sec:gen-approx-alg}

Our approximation algorithm uses the following fact.

\begin{lemma} \label{lem:one-vector-bound}
    Let $n,k \in \Z^+, k \le n$, let $q$ be a prime power, and let $\C_1,\C_2$ be $[n,k]_q$ codes with successive minima bases $G_1 = (\vec{v}_1,\dots,\vec{v}_k), G_2 = (\vec{w}_1,\dots,\vec{w}_k)$, respectively, and let $T$ be a linear transformation such that $TG_1 = G_2$. Furthermore, let $\vec{x} = G_1 \vec{a}$ for $\vec{a} \in \F_q^k$ be a non-zero codeword in $\C_1$.
    Suppose that $j \in [k]$ is the maximum index such that $a_j \neq 0$.
    Then
    $$\frac{\norm{T\vec{x}}_0}{\norm{\vec{x}}_0} \le \frac{\sum_{i=1}^j\lambda_i(\C_2)}{\lambda_j(\C_1)} \ \text{.}$$
\end{lemma}

\begin{proof}
Since $TG_1 = G_2$, we know that $T\vec{v}_i = \vec{w}_i$ for $1 \le i \le k$ and that $T\vec{x} = \sum_{i=1}^j T(a_i\vec{v}_i) = \sum_{i=1}^j a_i\vec{w}_i$.  Using the triangle inequality, we get the upper bound $\norm{T\vec{x}}_0 \le \sum_{i=1}^j \norm{a_i\vec{w}_i}_0 \le \sum_{i=1}^j \norm{\vec{w}_i}_0 = \sum_{i=1}^j \lambda_i(\C_2)$.

On the other hand, the $j$ vectors $\vec{x}$ and $\vec{v}_1,\dots,\vec{v}_{j-1}$ must be linearly independent because $a_j \neq 0$ and $\vec{v}_j \notin \lspan(\vec{v}_1,\dots,\vec{v}_{j-1})$. By the definition of successive minima, we then have that $\norm{\vec{x}}_0 \ge \lambda_j(\C_1)$. Combining this lower bound on $\norm{\vec{x}}_0$ with the upper bound on $\norm{T\vec{x}}_0$ above implies the claim.
\end{proof}

We define the following quantity relating the successive minima of $[n, k]_q$ codes $\C_1$ and $\C_2$:
\begin{equation} \label{eq:max-succ-min-ratio}
    M(\C_1,\C_2) \defeq \max_{i \in [k]}\frac{\lambda_i(\C_2)}{\lambda_i(\C_1)} \ \text{.}
\end{equation}
We will use this quantity to give lower and upper bounds on the distortion $\D(\C_1, \C_2)$.

\begin{theorem} \label{thm:approx-factor}
    Let $n,k \in \Z^+, k \le n$, let $q$ be a prime power, and let $\C_1,\C_2$ be $[n,k]_q$ codes. 
    Let $G_1 = (\vec{v}_1, \ldots, \vec{v}_k)$ and $G_2 = (\vec{w}_1, \ldots, \vec{w}_k)$ be successive minima bases of $\C_1$ and $\C_2$, respectively, and let $T$ be a linear transformation such that $T G_1 = G_2$
    Then
    $$M(\C_1,\C_2)M(\C_2,\C_1) \le \D(\C_1,\C_2) \leq \D_T(\C_1, \C_2) \le k^2 \cdot M(\C_1,\C_2)M(\C_2,\C_1) \ \text{.}$$
\end{theorem}

\begin{proof}
We first prove the upper bound.
For any $\vec{x} \in \C_1 \backslash \{\vec{0}\}$, we can write $\vec{x} = G\vec{a}$ for some $\vec{a} \in \F_q^k \setminus \set{\vec{0}}$.
Let $j \in [k]$ be the maximum index such that $a_j \neq 0$. 
By \cref{lem:one-vector-bound}, we then have that
\begin{align*}
    \frac{\norm{T\vec{x}}_0}{\norm{\vec{x}}_0} \leq \frac{\sum_{i=1}^j\lambda_i(\C_2)}{\lambda_j(\C_1)}
    \le \frac{\sum_{i=1}^j\lambda_j(\C_2)}{\lambda_j(\C_1)}
    \le k \cdot \frac{\lambda_j(\C_2)}{\lambda_j(\C_1)}
    \le k \cdot \max_{i \in [k]}\frac{\lambda_i(\C_2)}{\lambda_i(\C_1)} \ \text{.}
\end{align*}
As this is true for all $x \in \C_1 \backslash \{\vec{0}\}$, it holds that
$$\max_{\vec{x} \in \C_1 \setminus \set{\vec{0}}} \frac{\norm{T\vec{x}}_0}{\norm{\vec{x}}_0} \le k \cdot \max_{i \in [k]}\frac{\lambda_i(\C_2)}{\lambda_i(\C_1)} = k \cdot M(\C_1,\C_2) \ \text{.} $$
A symmetric argument applied to $T^{-1}$ yields
\[
\max_{\vec{y} \in \C_2 \setminus \set{\vec{0}}} 
\frac{\norm{T^{-1} \vec{y}}_0}{\norm{\vec{y}}_0} \le k \cdot M(\C_2,\C_1) \ \text{,}
\]
and plugging both into the definition of distortion yields the upper bound
$$\D(\C_1,\C_2) \le k^2 \cdot M(\C_1,\C_2)M(\C_2,\C_1) \ \text{.}$$

We next prove the lower bound using an analogous proof to one in~\cite{conf/esa/BennettDS16}.
Let $U$ be an arbitrary full-rank linear map such that $U \C_1 = \C_2$.  
Then, because $\vec{v}_1, \ldots, \vec{v}_i$ are linearly independent for every $i \in [k]$, $U \vec{v}_1,\dots, U \vec{v}_i$ are also linearly independent. It therefore follows that
\[
\lambda_i(\C_2) \le \max_{j \in [i]} \norm{U \vec{v}_j}_0 \le \norm{U|_{\C_1}}_{0 \rightarrow 0} \cdot \max_{j \in [i]}\norm{\vec{v}_j}_0 = \norm{U|_{\C_1}}_{0 \rightarrow 0}\lambda_i(\C_1) \ \text{.}
\]
This yields the bound $\frac{\lambda_i(\C_2)}{\lambda_i(\C_1)} \le \norm{U|_{\C_1}}_{0 \rightarrow 0}$ for all $i$, which implies $\max_{i \in [k]}\frac{\lambda_i(\C_2)}{\lambda_i(\C_1)} \le \norm{U|_{\C_1}}_{0 \rightarrow 0}$.  A symmetrical argument with $(U^{-1})|_{\C_2}$ mapping successive minima vectors of $\C_2$ to $\C_1$ yields $\max_{i \in [n]}\frac{\lambda_i(\C_1)}{\lambda_i(\C_2)} \le \norm{U^{-1}|_{\C_2}}_{0 \rightarrow 0}$, and multiplying the bounds together shows that
$$
\left(\max_{i \in [n]}\frac{\lambda_i(\C_2)}{\lambda_i(\C_1)}\right)\left(\max_{i \in [n]}\frac{\lambda_i(\C_1)}{\lambda_i(\C_2)}\right) \le \norm{U|_{\C_1}}_{0 \rightarrow 0}\norm{U^{-1}|_{\C_2}}_{0 \rightarrow 0}.
$$
Because $U$ was chosen arbitrarily, we then have that
$M(\C_1,\C_2)M(\C_2,\C_1) \le \D(\C_1,\C_2)$.
\end{proof}

We now restate and prove \cref{thm:intro-approx-alg}, which gives our main approximation algorithm for $\CDP_q$.

\approxalg*

\begin{proof}
    The algorithm works by computing successive minima bases $G_1$ and $G_2$ for $\C_1$ and $\C_2$, respectively, using \cref{lem:succ-min-basis}, and then computing and outputting a linear transformation $T$ such that $T G_1 = G_2$.
    By \cref{lem:succ-min-basis}, computing $G_1$ and $G_2$ takes $O^*(q^k)$ time and polynomial space, and using these bases it is efficient to compute $T$. 
    Furthermore, $\D_T(\C_1, \C_2) \leq k^2 \cdot \D(\C_1, \C_2)$ by \cref{thm:approx-factor}. The theorem follows.
\end{proof}

\full{
\subsection{An Improved Approximation Algorithm when \texorpdfstring{$\lambda_1 = \lambda_k$}{the Successive Minima are all the Same} with Tight Analysis}
\label{sec:approx-alg-all-succ-min-same}

In this section, we give an improved analysis of the approximation algorithm for $\CDP$ in \cref{sec:gen-approx-alg} in the case when the input codes $[n, k]_q$ codes are binary (i.e., when $q = 2$) and all of their successive minima are the same. See \cref{thm:approx-alg-succ-min-same-F2}.
Specifically, with this restriction we get an approximation factor of $\gamma = \big(\frac{2k + 1}{3}\big)^2 \approx 4k^2/9$, which is better than the $\gamma = k^2$ approximation factor that we achieve for general codes in \cref{sec:gen-approx-alg}. (The restriction to $q = 2$ is not inherent, but the proof is more complicated for larger values of $q$; see~\cite{morrell2026adventures} for a proof for general $q$.)
We also show that our analysis is tight in \cref{thm:tight-approx}.

We will use the following elementary claim.
\begin{claim} \label{clm:intersect-union-lb}
For sets $A, B \subseteq [k]$, $\card{A \cap B} \geq \card{A} + \card{B} - k$.
\end{claim}

\begin{proof}
We have that $\card{A} + \card{B} = \card{A \cup B} + \card{A \cap B} \leq k + \card{A \cap B}$. Subtracting $k$ from both sides implies the claim.
\end{proof}

We will also use the following simple fact about successive minima bases.
\begin{claim} \label{clm:triang-ineq-succ-min-bas}
Let $G = (\vec{g}_1, \ldots, \vec{g}_k) \in \F_q^{n \times k}$ be a successive minima basis of an $[n, k]_q$ code $\C$. Then, for all $\vec{m} \in \F_q^k$, $\norm{G \vec{m}}_0 \leq \lambda_k(\C) \cdot \norm{\vec{m}}_0$.
\end{claim}

\begin{proof}
Fix $\vec{m} \in \F_q^k$. Then, by the triangle inequality and the fact that $\norm{\vec{g}_i}_0 = \lambda_i(\C) \leq \lambda_k(\C)$,
\[
\norm{G\vec{m}}_0
= \Big\|\sum_{i = 1}^k m_i \vec{g}_i \Big\|_0
\leq \sum_{i=1}^k \norm{m_i \vec{g}_i}_0
= \sum_{i \in \supp(\vec{m})} \norm{\vec{g}_i}_0
\leq \sum_{i \in \supp(\vec{m})} \lambda_k(\C)
= \lambda_k(\C) \cdot \norm{\vec{m}}_0 \ \text{.} \qedhere
\]
\end{proof}

We now give our improved analysis of the approximation algorithm for $\CDP$.

\begin{theorem} \label{thm:approx-alg-succ-min-same-F2}
    Let $n,k \in \Z^+, k \le n$, and let $\C_1,\C_2$ be $[n,k]_2$ codes with successive minima bases $G_1 = (\vec{v}_1, \ldots, \vec{v}_k), G_2 = (\vec{w}_1, \ldots, \vec{w}_k) \in \F_2^{n \times k}$, respectively, and let $T$ be a linear map such that $T G_1 = G_2$. Furthermore, assume that $\lambda_1(\C_1) = \lambda_1(\C_2) = \lambda_k(\C_1) = \lambda_k(\C_2)$, and let $\lambda := \lambda_1(\C_1)$. Then
    $$\D(\C_1,\C_2) \leq \D_T(\C_1,\C_2) \le \left(\frac{2k+1}{3}\right)^2 \ \text{.}$$
\end{theorem}
\begin{proof}

    Let $T$ be a linear map such that $TG_1 = G_2$ (and therefore $T \C_1 = \C_2)$.
    Then
    \begin{equation} \label{eq:mixed-message-dist-bound}
    \D(\C_1,\C_2) \le \D_T(\C_1,\C_2) \le \max_{\vec{a} \in \F_2^k}\left(\frac{\norm{G_2\vec{a}}_0}{\norm{G_1\vec{a}}_0}\right)\max_{\vec{b} \in \F_2^k}\left(\frac{\norm{G_1\vec{b}}_0}{\norm{G_2\vec{b}}_0}\right) = \max_{\vec{a},\vec{b} \in \F_2^k}\left(\frac{\norm{G_1\vec{b}}_0}{\norm{G_1\vec{a}}_0}\frac{\norm{G_2\vec{a}}_0}{\norm{G_2\vec{b}}_0}\right) \ \text{.}
    \end{equation}

    Fix a pair of messages $\vec{a}, \vec{b} \in \F_2^k \setminus \set{\vec{0}}$. 
    We will upper bound the quantity $\frac{\norm{G_2\vec{a}}_0\norm{G_1\vec{b}}_0}{\norm{G_1\vec{a}}_0\norm{G_2\vec{b}}_0}$ appearing in the right-hand side of \cref{eq:mixed-message-dist-bound}
    by splitting into cases according to the average of $\norm{\vec{a}}_0$ and $\norm{\vec{b}}_0$.

    \begin{caseof}
        \case{$\frac{\norm{\vec{a}}_0 + \norm{\vec{b}}_0}{2} \le \frac{2k+1}{3}$.}
        In this case, we can upper bound the numerator and lower bound the denominator of $\frac{\norm{G_2\vec{a}}_0\norm{G_1\vec{b}}_0}{\norm{G_1\vec{a}}_0\norm{G_2\vec{b}}_0}$ separately. 
        We first upper bound the numerator:
        \begin{align*}
            \norm{G_2\vec{a}}_0\norm{G_1\vec{b}}_0 
             \leq (\lambda\norm{\vec{a}}_0)(\lambda\norm{\vec{b}}_0) 
             \le \left(\frac{\lambda\norm{\vec{a}}_0 + \lambda\norm{\vec{b}}_0}{2}\right)^2
            \le \lambda^2 \cdot \Big(\frac{2k+1}{3}\Big)^2  \ \text{.}
        \end{align*}
        The first inequality uses \cref{clm:triang-ineq-succ-min-bas} twice, the second inequality uses the AM-GM inequality, and the third inequality uses the case assumption.
        We then have that
        \begin{align*}
            \frac{\norm{G_2\vec{a}}_0\norm{G_1\vec{b}}_0}{\norm{G_1\vec{a}}_0\norm{G_2\vec{b}}_0} 
            \le \frac{\lambda^2 \cdot \big(\frac{2k+1}{3} \big)^2}{\norm{G_1\vec{a}}_0\norm{G_2\vec{b}}_0} 
            \le \frac{\lambda^2 \cdot \big(\frac{2k+1}{3} \big)^2}{\lambda^2}
            = \Big(\frac{2k+1}{3}\Big)^2 \ \text{,}
        \end{align*}
        where the second inequality holds because the minimum distance of both $\C_1$ and $\C_2$ is $\lambda$.
        
        \case{$\frac{\norm{\vec{a}}_0 + \norm{\vec{b}}_0}{2} > \frac{2k+1}{3}$.}
        In this case, we will upper bound $\frac{\norm{G_1\vec{b}}_0}{\norm{G_1\vec{a}}_0}$ and $\frac{\norm{G_2\vec{a}}_0}{\norm{G_2\vec{b}}_0}$ separately to get a bound on their product $\frac{\norm{G_2\vec{a}}_0\norm{G_1\vec{b}}_0}{\norm{G_1\vec{a}}_0\norm{G_2\vec{b}}_0}$.
        To do this, we introduce a new message $\vec{x}$ defined as the unique element $\vec{x} \in \F_2^k$ such that $\supp(\vec{x}) = \supp(\vec{a}) \cap \supp(\vec{b})$ (equivalently, $\vec{x}$ is the bit-wise AND of $\vec{a}$ and $\vec{b}$).
        Additionally, let $\vec{a}' := \vec{a} - \vec{x}$ and $\vec{b}' := \vec{b} - \vec{x}$.
        Note that $\supp(\vec{a}) \cup \supp(\vec{b}) \subseteq [k]$ is equal to the disjoint union $\supp(\vec{a}') \sqcup \supp(\vec{b}') \sqcup \supp(\vec{x})$. In particular, $\norm{\vec{a}'}_0 + \norm{\vec{b}'}_0 + \norm{\vec{x}}_0 \le k$.
        We will now show that $\frac{\norm{G_1\vec{b}}_0}{\norm{G_1\vec{a}}_0} < \frac{2k+1}{3}$.

        First, we must note that by the triangle inequality,
        \begin{equation} \label{eq:G1ba-lb}
            \norm{G_1\vec{b}}_0 - \norm{G_1\vec{a}}_0 \le \norm{G_1\vec{b} - G_1\vec{a}}_0
            \le \norm{G_1(\vec{b} - \vec{a})}_0
            \le \lambda \cdot \norm{\vec{b} - \vec{a}}_0 \ \text{.}
        \end{equation}
        By invoking \cref{clm:intersect-union-lb} with $A := \supp(\vec{a})$ and $B := \supp(\vec{b})$ and using the case lower bound, we have that
        \[
        \norm{\vec{x}}_0 \geq \norm{\vec{a}}_0 + \norm{\vec{b}}_0 - k > \frac{4k + 2}{3} - k = \frac{k + 2}{3} \ \text{.}
        \]
        So, using the fact that $\vec{a}, \vec{b}$ are vectors over $\F_2$, and the definitions of $\vec{a}', \vec{b}'$,
        \begin{equation} \label{eq:ba-diff-lb}
        \norm{\vec{b} - \vec{a}}_0 = \norm{\vec{a}' + \vec{b}'}_0 \leq k - \norm{\vec{x}}_0
        < k - \frac{k + 2}{3} = \frac{2k - 2}{3}
        \ \text{.}
        \end{equation}
        By combining \cref{eq:G1ba-lb,eq:ba-diff-lb},
        \[
        \norm{G_1 \vec{b}}_0 < \frac{\lambda (2k - 2)}{3} + \norm{G_1 \vec{a}}_0 \ \text{,}
        \]
        and so, using that $\norm{G_1 \vec{a}}_0 \geq \lambda$,
        \[
        \frac{\norm{G_1 \vec{b}}_0}{\norm{G_1 \vec{a}}_0} < \frac{\lambda (2k - 2)}{3 \norm{G_1 \vec{a}}_0} + 1 < \frac{2k - 2}{3} + 1 = \frac{2k + 1}{3} \ \text{.}
        \]
        
Essentially the same analysis shows that $\frac{\norm{G_2\vec{a}}_0}{\norm{G_2\vec{b}}_0} < (2k + 1)/3$. Therefore,
$$
\frac{\norm{G_2\vec{a}}_0\norm{G_1\vec{b}}_0}{\norm{G_1\vec{a}}_0\norm{G_2\vec{b}}_0} < \left(\frac{2k+1}{3}\right)^2 \ \text{,}
$$
as needed.
    \end{caseof}
\end{proof}

\approxalgsuccmin*

\begin{proof}
    The algorithm works by computing successive minima bases $G_1$ and $G_2$ for $\C_1$ and $\C_2$, respectively, using \cref{lem:succ-min-basis}, and then computing and outputting a linear transformation $T$ such that $T G_1 = G_2$.
    By \cref{lem:succ-min-basis}, computing $G_1$ and $G_2$ takes $O^*(q^k)$ time and polynomial space, and using these bases it is efficient to compute $T$. 
    Furthermore, by the lower bound in \cref{thm:approx-factor} and the upper bound in \cref{thm:approx-alg-succ-min-same-F2} (which applies because $\lambda_1(\C_1) = \cdots = \lambda_k(\C_1) = \lambda_1(\C_2) = \cdots = \lambda_k(\C_2)$), $\D_T(\C_1, \C_2) \leq \big(\frac{2k + 1}{3}\big)^2 \cdot \D(\C_1, \C_2)$, as needed.
\end{proof}

Moreover, we show that the analysis in \cref{thm:approx-alg-succ-min-same-F2} is \emph{tight}. That is, we give a pair of generator matrices $G_1, G_2$ for $[n, k]_2$ codes $\C_1, \C_2$ all of whose successive minima are the same such that any linear map $T$ with $T \C_1 = \C_2$ has distortion $\D_T(\C_1, \C_2) = \big(\frac{2k + 1}{3}\big)^2$. In fact, we simply choose $G_1, G_2$ to be different generator matrices of the same code $\C$.

\begin{theorem} \label{thm:tight-approx}
    Let $n,k,\ell \in \Z^+, \ell \ge 2, k:= 3\ell + 1, n := (2\ell + 1)\ell$. There exist $[n,k]_2$ codes $\C_1,\C_2$ with successive minima bases $G_1 = (\vec{v}_1,\dots,\vec{v}_k)$, $G_2 = (\vec{w}_1,\dots,\vec{w}_k)$, respectively, satisfying:
    \begin{enumerate}
        \item \label{item:all-succ-min-equal-example} $\ell = \lambda_1(\C_1) = \lambda_1(\C_2) = \lambda_k(\C_1) = \lambda_k(\C_2)$.
        \item The distortion between $\C_1$ and $\C_2$ is $\D(\C_1,\C_2) = 1.$
        \item For any linear transformation $T \in \F_2^{n\times n}$ such that $TG_1 = G_2$,
        $$\D_T(\C_1,\C_2) = \left(\frac{2k+1}{3}\right)^2.$$
    \end{enumerate}
\end{theorem}

In other words, an algorithm that picks arbitrary successive minima bases of $\C_1$, $\C_2$ to approximate $\D(\C_1,\C_2)$ must achieve an approximation factor of at least $\left(\frac{2k+1}{3}\right)^2$.

\begin{proof}
    We will give a constructive proof, defining $\C = \C_1 = \C_2$ by the bases $G_1,G_2$ which contain the same columns in a different order:
    \begin{align*}
        G_1 &:= \begin{pmatrix}
               A & B & \vec{0}_{(\ell^2+\ell)\times\ell}\\
               \vec{0}_{\ell^2 \times (\ell + 1)} & \vec{0}_{\ell^2 \times \ell} & C\\
        \end{pmatrix} \ \text{,} \\
        G_2 &:= \begin{pmatrix}
               A & \vec{0}_{(\ell^2+\ell)\times\ell} & B\\
               \vec{0}_{\ell^2 \times (\ell + 1)} & C & \vec{0}_{\ell^2 \times \ell}
        \end{pmatrix}  \ \text{.}
    \end{align*}
    Here $A \in \F_2^{(\ell^2 + \ell) \times (\ell + 1)}$, $B \in \F_2^{(\ell^2 + \ell) \times \ell}$, $C \in \F_2^{\ell^2 \times \ell}$ are defined as
    \begin{align*}
        A &:= I_{\ell + 1} \otimes \vec{1}_{\ell} \ \text{,} \\
        B &:= \vec{1}_{\ell} \otimes \begin{pmatrix}
            \vec{0}_{1 \times \ell}\\
            I_\ell
        \end{pmatrix} \ \text{,} \\
        C &:= I_\ell \otimes \vec{1}_{\ell} \ \text{.}
    \end{align*}
    Here $A$ is the identity matrix $I_{\ell+1}$ ``scaled by $\ell$ in unary,'' $C$ is a similar scaling of $I_\ell$, and $B$ is several repetitions of a row of $0$s followed by $I_\ell$.  Importantly, this set-up ensures that all columns of $G_1$ are linearly independent (and likewise for $G_2$), and that all columns of $G_1,G_2$ have Hamming weight $\ell$.
    We next prove the following claim.

    \begin{claim} \label{clm:lambda1-ex}
    \cref{item:all-succ-min-equal-example} holds. 
    \end{claim}
    \begin{proof}
        Note that because each of the columns of $G_1, G_2$ has Hamming weight $\ell$, it suffices to show that $\lambda_1(\C) = \ell$. We do this by confirming that the columns of $A$ and $B$ cannot be combined into shorter codewords.  Observe that any two columns of $A$ have disjoint supports, as do any two columns of either $B$ or $C$.  Additionally, the columns of $C$ have disjoint supports from the columns of both $A$ and $B$, so the only codewords that could possibly yield a hamming weight less than $\ell$ are ones of the form $A\vec{x}_A + B\vec{x}_B$ for some $\vec{x}_A \in \F_2^{\ell+1},\vec{x}_B \in \F_2^{\ell}$ with $\vec{x}_A \neq \vec{0},\vec{x}_B \neq \vec{0}$.

        We now partition the $\ell(\ell+1)$ coordinates of $A$ and $B$ into $\ell+1$ blocks of $\ell$ consecutive coordinates.  In $A$, each block has one column with $\ell$ $1$s with all other columns being $\vec{0}$.  Meanwhile, each of the first $\ell$ blocks in $B$ has exactly one row and one column that is all $0$s, and $\ell-1$ $1$s spread across the remaining $\ell-1$ columns.  The last ($\ell+1$st) block in $B$ is $I_{\ell}$.  Note that any block in $A\vec{x}_A$ will either be $\vec{0}_{\ell}$ or $\vec{1}_{\ell}$, so any block in $B\vec{x}_B$ that is neither $\vec{1}_{\ell}$ or $\vec{0}_{\ell}$ must contain at least one $1$ in $A\vec{x}_A + B\vec{x}_B$.  We call these blocks of $B\vec{x}_B$ \emph{mixed}.

        To prove the claim, we will show that $B\vec{x}_B$ must have at least $\ell$ mixed blocks.  If $\norm{\vec{x}_B}_0 = \ell$ (so all columns of $B$ are included), each of the first $\ell$ blocks contains exactly $\ell-1$ ones and a zero, providing the $\ell$ mixed blocks we require.  When $0 < \norm{\vec{x}_B}_0 < \ell$, the $\ell+1$st block contains $\norm{\vec{x}_B}_0$ ones and is mixed, and we can pick any included column of $B$ to find the remaining mixed blocks.  Each column has the rest of its $\ell-1$ support spread across $\ell-1$ different blocks among the first $\ell$, and each of those blocks has a row of zeroes guaranteeing that it is mixed if it contains any ones.  Thus we have $\ell$ total mixed blocks ensuring that $\norm{A\vec{x}_A + B\vec{x}_B}_0 \ge \ell$ for any $\vec{x}_A,\vec{x}_B$ and $\lambda_1(\C) = \ell$ as a result. 
        \end{proof}

    Now that we have established that $\lambda_1(\C) = \lambda_k(\C) = \ell$, consider the following two messages $\vec{a},\vec{b} \in \F_2^k$:
    \begin{align*}
    \vec{a} &= \begin{pmatrix}\vec{1}_{\ell+1}\\\vec{1}_{\ell}\\\vec{0}_{\ell}\end{pmatrix} \ \text{,} \\
    \vec{b} &= \begin{pmatrix}\vec{1}_{\ell+1}\\\vec{0}_{\ell}\\\vec{1}_{\ell}\end{pmatrix} \ \text{.}
    \end{align*}
    Note that $G_1\vec{a} = G_2\vec{b}$ and $G_1\vec{b}= G_2\vec{a}$, and these can be computed as
    $$
        G_1\vec{a} = G_2\vec{b} = \begin{pmatrix}A\\
               \vec{0}_{\ell^2 \times (\ell + 1)}\\\end{pmatrix}\vec{1}_{\ell+1} + \begin{pmatrix}B\\
               \vec{0}_{\ell^2 \times \ell}\\\end{pmatrix}\vec{1}_{\ell} = \begin{pmatrix}\vec{1}_{\ell}\otimes\begin{pmatrix}1\\\vec{1}_{\ell}\\\end{pmatrix}\\ \vec{0}_{\ell^2}\end{pmatrix} + \begin{pmatrix}\vec{1}_{\ell}\otimes\begin{pmatrix}0\\\vec{1}_{\ell}\\\end{pmatrix}\\ \vec{0}_{\ell^2}\end{pmatrix} = \begin{pmatrix}\vec{1}_{\ell}\otimes\begin{pmatrix}1\\\vec{0}_{\ell}\\\end{pmatrix}\\ \vec{0}_{\ell^2}\end{pmatrix}
               $$
               and
        $$       
        G_1\vec{b} = G_2\vec{a} = \begin{pmatrix}A\\
               \vec{0}_{\ell^2 \times (\ell + 1)}\\\end{pmatrix}\vec{1}_{\ell+1} + \begin{pmatrix}\vec{0}_{(\ell^2+\ell) \times \ell} \\
               C\\\end{pmatrix}\vec{1}_{\ell} = \begin{pmatrix}\vec{1}_{\ell^2 + \ell}\\\vec{0}_{\ell^2}\end{pmatrix} + \begin{pmatrix}\vec{0}_{\ell^2 + \ell}\\\vec{1}_{\ell^2}\end{pmatrix} = \vec{1}_{2\ell^2 + \ell} \ \text{.}
    $$
    From this, we can conclude that any $T \in \F_2^{n \times n}$ such that $TG_1 = G_2$ will have distortion on $\C$ of at least
    \begin{align*}
        \D_T(\C,\C) &\ge \frac{\norm{TG_1\vec{a}}_0}{\norm{G_1\vec{a}}_0}\frac{\norm{G_1\vec{b}}_0}{\norm{TG_1\vec{b}}_0}\\
        &= \frac{\norm{G_2\vec{a}}_0}{\norm{G_1\vec{a}}_0}\frac{\norm{G_1\vec{b}}_0}{\norm{G_2\vec{b}}_0}\\
        &= \left(\frac{2\ell^2+\ell}{\ell}\right)\left(\frac{2\ell^2+\ell}{\ell}\right)\\
        &= \left(2\ell + 1\right)^2.
    \end{align*}
    Since $k = 3\ell + 1$, we have $2\ell + 1  = \frac{6\ell + 3}{3} = \frac{2(3\ell+1) + 1}{3} = \frac{2k + 1}{3}$ so
    $$\D_T(\C,\C) \ge \left(2\ell + 1\right)^2 = \left(\frac{2k+1}{3}\right)^2.$$
    Finally, we know that the distortion between a code and itself is necessarily $\D(\C,\C) = 1$, completing the proof.
\end{proof}
}

\bibliography{code-distortion,equivalence}

\newcommand{\etalchar}[1]{$^{#1}$}
\begin{thebibliography}{DDvW22}

\bibitem[ABC{\etalchar{+}}22]{classic-mceliece}
Martin~R. Albrecht, Daniel~J. Bernstein, Tung Chou, Carlos Cid, Jan Gilcher,
  Tanja Lange, Varun Maram, Ingo von Maurich, Rafael Misoczki, Ruben
  Niederhagen, Kenneth~G. Paterson, Edoardo Persichetti, Christiane Peters,
  Peter Schwabe, Nicolas Sendrier, Jakub Szefer, Cen~Jung Tjhai, Martin
  Tomlinson, , and Wen Wang.
\newblock {Classic McEliece}, 2022.
\newblock {NIST} Post-Quantum Cryptography Standardization Project submission.

\bibitem[ABSS97]{journals/jcss/AroraBSS97}
Sanjeev Arora, L{\'{a}}szl{\'{o}} Babai, Jacques Stern, and Z.~Sweedyk.
\newblock The hardness of approximate optima in lattices, codes, and systems of
  linear equations.
\newblock {\em J. Comput. Syst. Sci.}, 54(2):317--331, 1997.
\newblock Preliminary version in FOCS 1993.

\bibitem[AK14]{journals/tit/AustrinK14}
Per Austrin and Subhash Khot.
\newblock A simple deterministic reduction for the gap minimum distance of code
  problem.
\newblock {\em {IEEE} Trans. Inf. Theory}, 2014.
\newblock Preliminary verison in ICALP 2011.

\bibitem[BBB{\etalchar{+}}26]{BBB+asymp-improvements26}
Huck Bennett, Drisana Bhatia, Jean-François Biasse, Medha Durisheti, Lucas
  LaBuff, Vincenzo Pallozzi~Lavorante, and Philip Waitkevich.
\newblock Asymptotic improvements to provable algorithms for the code
  equivalence problem.
\newblock {\em IEEE Transactions on Information Theory}, 72(2):1093--1108,
  2026.

\bibitem[BBPS21]{conf/pqcrypto/BarenghiBPS21}
Alessandro Barenghi, Jean{-}Fran{\c{c}}ois Biasse, Edoardo Persichetti, and
  Paolo Santini.
\newblock {LESS-FM:} fine-tuning signatures from the code equivalence problem.
\newblock In {\em PQCrypto}, 2021.

\bibitem[BBPS23]{journals/amco/BarenghiBPS23}
Alessandro Barenghi, Jean{-}Fran{\c{c}}ois Biasse, Edoardo Persichetti, and
  Paolo Santini.
\newblock On the computational hardness of the code equivalence problem in
  cryptography.
\newblock {\em Adv. Math. Commun.}, 17(1):23--55, 2023.

\bibitem[BCGQ11]{conf/soda/BabaiCGQ11}
L{\'{a}}szl{\'{o}} Babai, Paolo Codenotti, Joshua~A. Grochow, and Youming Qiao.
\newblock Code equivalence and group isomorphism.
\newblock In {\em SODA}, 2011.

\bibitem[BDS16]{conf/esa/BennettDS16}
Huck Bennett, Daniel Dadush, and Noah Stephens{-}Davidowitz.
\newblock On the lattice distortion problem.
\newblock In {\em ESA}, 2016.

\bibitem[Beu20]{conf/sacrypt/Beullens20}
Ward Beullens.
\newblock Not enough {LESS:} an improved algorithm for solving code equivalence
  problems over $\mathbb{F}_q$.
\newblock In {\em SAC}, 2020.

\bibitem[BGG{\etalchar{+}}19]{conf/soda/BhattiproluGGLT19}
Vijay Bhattiprolu, Mrinalkanti Ghosh, Venkatesan Guruswami, Euiwoong Lee, and
  Madhur Tulsiani.
\newblock Approximability of $p \to q$ matrix norms: Generalized krivine
  rounding and hypercontractive hardness.
\newblock In {\em SODA}, 2019.

\bibitem[BMPS20]{conf/africacrypt/BiasseMPS20}
Jean{-}Fran{\c{c}}ois Biasse, Giacomo Micheli, Edoardo Persichetti, and Paolo
  Santini.
\newblock {LESS} is more: Code-based signatures without syndromes.
\newblock In {\em AFRICACRYPT}, volume 12174, pages 45--65. Springer, 2020.

\bibitem[BW24]{BennettWin2024}
Huck Bennett and Kaung Myat~Htay Win.
\newblock Relating code equivalence to other isomorphism problems.
\newblock {\em Designs, Codes and Cryptography}, 93(3):701–723, Dec 2024.

\bibitem[CW12]{journals/tit/ChengW12}
Qi~Cheng and Daqing Wan.
\newblock A deterministic reduction for the gap minimum distance problem.
\newblock {\em {IEEE} Trans. Inf. Theory}, 58(11):6935--6941, 2012.
\newblock Preliminary version in STOC 2009.

\bibitem[DDvW22]{journals/tit/Debris-AlazardD22}
Thomas Debris{-}Alazard, L{\'{e}}o Ducas, and Wessel P.~J. van Woerden.
\newblock An algorithmic reduction theory for binary codes: {LLL} and more.
\newblock {\em {IEEE} Trans. Inf. Theory}, 68(5):3426--3444, 2022.

\bibitem[DG23]{conf/pkc/DucasG23}
L{\'{e}}o Ducas and Shane Gibbons.
\newblock Hull attacks on the lattice isomorphism problem.
\newblock In {\em PKC}, 2023.

\bibitem[DMS03]{journals/tit/DumerMS03}
Ilya Dumer, Daniele Micciancio, and Madhu Sudan.
\newblock Hardness of approximating the minimum distance of a linear code.
\newblock {\em {IEEE} Trans. Inf. Theory}, 49(1):22--37, 2003.
\newblock Preliminary version in FOCS 1999.

\bibitem[GAA{\etalchar{+}}25]{hqc_spec_2025}
Philippe Gaborit, Carlos {Aguilar Melchor}, Nicolas Aragon, Slim Bettaieb,
  Lo{\"i}c Bidoux, Olivier Blazy, Jean{-}Christophe Deneuville, Edoardo
  Persichetti, Gilles Z{\'e}mor, Jurjen Bos, Arnaud Dion, J{\'e}r{\^o}me Lacan,
  Jean{-}Marc Robert, Pascal V{\'e}ron, Paulo~L. Barreto, Santosh Ghosh, Shay
  Gueron, Tim G{\"u}neysu, Rafael Misoczki, Jan Richter{-}Brokmann, Nicolas
  Sendrier, Jean{-}Pierre Tillich, and Valentin Vasseur.
\newblock {HQC} cryptosystem specification.
\newblock \url{https://pqc-hqc.org/doc/hqc_specifications_2025_08_22.pdf},
  2025.
\newblock Version dated 2025-08-22.

\bibitem[GMSS99]{journals/ipl/GoldreichMSS99}
Oded Goldreich, Daniele Micciancio, Shmuel Safra, and Jean{-}Pierre Seifert.
\newblock Approximating shortest lattice vectors is not harder than
  approximating closest lattice vectors.
\newblock {\em Inf. Process. Lett.}, 71(2):55--61, 1999.

\bibitem[GS24]{conf/approx/GhentiyalaS24}
Surendra Ghentiyala and Noah Stephens{-}Davidowitz.
\newblock More basis reduction for linear codes: Backward reduction, {BKZ},
  slide reduction, and more.
\newblock In {\em APPROX}, 2024.

\bibitem[GS26]{golovnev-stephens-davidowitz-personal-comm-26}
Alexander Golovnev and Noah Stephens{-}Davidowitz.
\newblock Personal communication, 2026.

\bibitem[Leo82]{Leon82}
J.~Leon.
\newblock Computing automorphism groups of error-correcting codes.
\newblock {\em IEEE Transactions on Information Theory}, 28(3):496--511, 1982.

\bibitem[Mat13]{Matousek2013MetricEmbeddings}
Ji{\v{r}}{\'i} Matou{\v{s}}ek.
\newblock Lecture notes on metric embeddings.
\newblock Lecture notes, Charles University, 2013.
\newblock Available at \url{https://kam.mff.cuni.cz/~matousek/ba-a4.pdf}.

\bibitem[McE78]{mceliece78}
Robert~J. McEliece.
\newblock A public-key cryptosystem based on algebraic coding theory, 1978.
\newblock {DSN} Progress Report.

\bibitem[Mic14]{conf/coco/Micciancio14}
Daniele Micciancio.
\newblock Locally dense codes.
\newblock In {\em CCC}, 2014.

\bibitem[Mor26]{morrell2026adventures}
Bryant Morrell.
\newblock Adventures with code distortion: Hardness and approximation.
\newblock Master's thesis, University of Colorado Boulder, 2026.

\bibitem[{Nat}25]{nist_pqc_round2_additional_signatures_2025}
{National Institute of Standards and Technology (NIST)}.
\newblock {Post-Quantum Cryptography: Additional Digital Signature Schemes —
  Round 2 Additional Signatures}.
\newblock
  \url{https://csrc.nist.gov/projects/pqc-dig-sig/round-2-additional-signatures},
  2025.

\bibitem[Now25]{Nowakowski25}
Julian Nowakowski.
\newblock An improved algorithm for code equivalence.
\newblock In {\em PQCrypto}, page 71–103, 2025.

\bibitem[PR97]{article/PetrankR97}
E.~Petrank and R.M. Roth.
\newblock Is code equivalence easy to decide?
\newblock {\em IEEE Transactions on Information Theory}, 43(5):1602--1604,
  1997.

\bibitem[Reg14]{regev14}
Oded Regev, 2014.
\newblock Personal communication.

\bibitem[Sen00]{journals/tit/Sendrier00}
Nicolas Sendrier.
\newblock Finding the permutation between equivalent linear codes: The support
  splitting algorithm.
\newblock {\em {IEEE} Trans. Inf. Theory}, 46(4):1193--1203, 2000.

\bibitem[TS17]{Ta-Shma2017}
Amnon Ta-Shma.
\newblock Explicit, almost optimal, epsilon-balanced codes.
\newblock In {\em Proceedings of the 49th Annual ACM SIGACT Symposium on Theory
  of Computing}, STOC 2017, page 238–251, New York, NY, USA, 2017.
  Association for Computing Machinery.

\end{thebibliography}
\bibliographystyle{alpha}

\end{document}